\documentclass{article}
\usepackage[utf8]{inputenc}
\usepackage{amsmath}
\usepackage{amsfonts}
\usepackage{amssymb}
\usepackage{amsthm}
\usepackage{thmtools}
\usepackage{thm-restate}
\usepackage{xcolor,graphicx}
\usepackage{tabularx,booktabs,multirow}
\usepackage[margin=0pt]{subcaption}
\usepackage{mathtools}
\usepackage{caption}
\usepackage[]{algorithm2e}
\usepackage{tikz}
\usetikzlibrary{arrows,decorations.pathmorphing,backgrounds,positioning,fit,matrix,shapes,calc,math}
\usepackage{hyperref}
\usepackage{cleveref}
\usepackage{todonotes}
\usepackage{a4wide}

\title{On Reachable Assignments in Cycles and Cliques}
\author{Luis Müller \and Matthias Bentert}
\date{}

\newtheorem{definition}{Definition}
\newtheorem{theorem}{Theorem}
\newtheorem{lemma}{Lemma}
\newtheorem{proposition}{Proposition}
\newtheorem{observation}{Observation}
\newtheorem{rrule}{Reduction Rule}
\newtheorem{construction}{Construction}
\newtheorem{example}{Example}

\newcommand{\problemdef}[3]{
		\begin{center}
	\begin{minipage}{0.95\textwidth}
		\noindent
		\textsc{#1}

				\vspace{2pt}
				\setlength{\tabcolsep}{3pt}
				\begin{tabularx}{\textwidth}{@{}lX@{}}
						\textbf{Input:} 		& #2 \\
						\textbf{Question:} 	& #3
					\end{tabularx}
	\end{minipage}
		\end{center}
}

\def \cycseq[#1][#2]{ \ensuremath \mu_{#1, #2} }
\def \ro{\textsc{Reachable Object}}
\def \ra{\textsc{Reachable Assignment}}
\def \fsra{\textsc{First Swap Reachable Assignment}}

\newcommand{\racycle}[3]{
\draw (#1,#2) circle (#3);
\draw[>=latex, double=gray,double distance = .6cm,line cap=round,opacity=.05] (#1,#2) circle (#3);
}

\newcommand{\objpath}[7]{
\tikzmath{\cs = cos(#4); \ss = sin(#4); \hxs =\cs * #3 ; \hys =\ss * #3; \xs = \hxs + #1 ; \ys =\hys +#2 ;}
\tikzmath{\ce = cos(#5); \se = sin(#5); \hxe =\ce * #3 ; \hye =\se * #3; \xe = \hxe + #1 ; \ye =\hye +#2 ;}
\node [draw, circle, fill=black, inner sep = 0.1cm] at (\xs ,\ys) {};
\node at (\hxs * 1.6 + #1, \hys * 1.4 + #2 ) {$\sigma_0^{-1}(#6)$};

\draw[>=latex, double=#7,double distance = .6cm,line cap=round,opacity=.25] (\xs ,\ys) arc (#4:#5:#3);

\node [draw, circle, fill=black, inner sep = 0.1cm] at (\xe ,\ye) {};
\node at (\hxe * 1.6 + #1, \hye * 1.4 + #2 ) {$\sigma^{-1}(#6)$};
}

\begin{document}

\maketitle

\begin{abstract}
The efficient and fair distribution of indivisible resources among agents is a common problem in the field of \emph{Multi-Agent-Systems}.
We consider a graph-based version of this problem called \ra{}, introduced by Gourves, Lesca, and Wilczynski [AAAI, 2017].
The input for this problem consists of a set of agents, a set of objects, the agent's preferences over the objects, a graph with the agents as vertices and edges encoding which agents can trade resources with each other, and an initial and a target distribution of the objects, where each agent owns exactly one object in each distribution.
The question is then whether the target distribution is reachable via a sequence of rational trades.
A trade is rational when the two participating agents are neighbors in the graph and both obtain an object they prefer over the object they previously held.
We show that \ra{} is NP-hard even when restricting the input graph to be a clique and develop an~$\mathcal{O}(n^3)$-time algorithm for the case where the input graph is a cycle with~$n$ vertices.
\end{abstract}

\section{Introduction}
The efficient distribution of resources among agents is a frequent problem in \emph{Multi-Agent-Systems}~\cite{DBLP:conf/wine/0001W19} with e.\,g.\ medical applications \cite{DBLP:conf/sigecom/AbrahamBS07}.
These resources are often modeled as objects and sometimes they can be divided among agents and sometimes they are indivisible.
One famous problem in this field is called \textsc{House Marketing}: Each of the~$n$ participating agents initially owns one house (an indivisible object) and the agents can trade their houses in trade cycles with other agents~\cite{roth_incentive_1982,shapley_cores_1974}.
Versions of this problem were considered with different optimization criteria like pareto-optimality~\cite{AbrCecManMeh2005,IP18,SoeUen2010} or envy-freeness \cite{BCGLMW18,BKN18,DBLP:conf/aaai/ChevaleyreEM07}.
Gourves et al.\,\cite{gourves_object_2017} studied two similar problems where agents are only able to perform (pairwise) trades with agents they trust.
This is modeled by a social network of the participating agents where an edge between two agents means that they trust each other.
The first version is called \ra: Can one reach a given target assignment by a sequence of rational swaps?
A swap is rational if both participating agents obtains an object they prefer over their current object and the agents share an edge in the social network.
The second version is called \ro{} and the question is whether there is a sequence of rational swaps such that a given agent obtains a given target object.

Gourves et al. \cite{gourves_object_2017} showed that \ra{} and \ro{} are both NP-hard on general graphs.
They further proved that \ra{} is decidable in polynomial time if~$G$ is a tree.
Huang and Xiao \cite{DBLP:conf/aaai/HuangX19} showed that if the underlying graph is a path, then \ro{} can be solved in polynomial time.
Moreover, they studied a version of \ro{} that allows weak preference lists, i.e. an agent can be indifferent between different objects and shows that this problem is NP-hard even if the input graph is a path.
Contributing to the \ro{} problem, Saffidine and Wilczynski \cite{DBLP:conf/sagt/SaffidineW18} proposed an alternative version of \ro, called \textsc{Guaranteed Level Of Satisfaction}, where an agent is guaranteed to obtain an object at least as good as a given target object.
They showed that \textsc{Guaranteed Level Of Satisfaction} is co-NP-hard.
Finally, Bentert et al.\,\cite{bentert_good_2019} showed that \ro{} is polynomial-time solvable if the input graph is a cycle and NP-hard even if the input graph is a clique, that is, all agents can trade with each other.

In our work, we will mainly focus on \ra{} for the two special cases where the input graph is either a cycle or a clique.
We show an~$\mathcal{O}(n^3)$-time algorithm for \ra{} on~$n$-vertex cycles and further prove NP-hardness for \ra{} on cliques.
\Cref{fig:cycleExample} displays an example for a triangle, that is, a cycle (and a clique) of size three.
\begin{figure}
    \centering

\begin{minipage}[c]{4cm}
\begin{tikzpicture}
   \node[circle, draw] (1) {$1$};
    \node[circle, draw] (3) [below left = 1.5cm and 1cm of 1]  {$3$};
    \node[circle, draw] (2) [below right = 1.5cm and 1cm of 1] {$2$};

    \path[draw,thick]
    (1) edge node {} (2)
    (2) edge node {} (3)
    (3) edge node {} (1);
\end{tikzpicture}
\end{minipage}
\begin{minipage}[c]{4cm}
1: $x_2$ $\succ$ \fbox{$x_1$}\\
2: $x_3$ $\succ$ \fbox{$x_2$}\\
3: $x_1$ $\succ$ $x_2$ $\succ$ \fbox{$x_3$}
\end{minipage}

    \caption{Example for \ra{} on a triangle with preference lists on the right-hand side.
    We use the notation ``1: $x_2$ $\succ$ \fbox{$x_1$}'' to denote that agent~$1$ prefers object~$x_2$ the most and object~$x_1$ the second most.
    Moreover, agent~$1$ initially holds object~$x_1$ and since we only consider rational swaps and hence agent~$x_1$ will never held an object it prefers less than object~$x_1$, we do not list these objects for agent~$1$.
    In the target assignment each agent shall hold its most preferred object.
    First, agents~$2$ and~$3$ can swap their currently held objects.
    Afterwards, agent~$1$ can trade object~$x_1$ to agent~$3$ and receive object~$x_2$ in return.
    }
    \label{fig:cycleExample}
\end{figure}
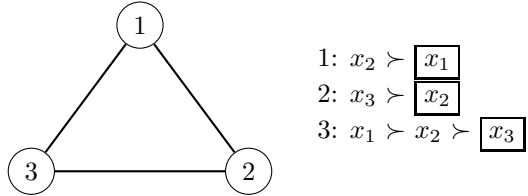

At the core of our hardness proof for cliques lies the combination of a reduction from \ro{} to \ra{} by Gouvres et al.\,\cite{gourves_object_2017} and the NP-hardness proof for \ro{} on cliques by Bentert et al.\,\cite{bentert_good_2019}.
The polynomial-time algorithm for \ra{} on cycles is shown in three steps.
In the first step, we will show that once an object is swapped into either clockwise or counter-clockwise direction, it is impossible to swap it back into the opposite direction.
Moreover, assigning a direction to each object is equivalent to providing a sequence of rational swaps.
We will also show how to verify in polynomial time whether such an assignment corresponds to a solution.
In this case, we will say that the assignment of directions \emph{yields the target assignment}.
In a second step, we show a characterization of assignments of directions that yield the target assignment.
We call these direction assignments \emph{valid}.
In the third and final step, we will iterate over all edges in the input graph and construct for each iteration a 2-SAT formula that is satisfiable if and only if there exists a valid assignment of directions that corresponds to a solution in which the first swap is done over the iterated edge.
Besides a novel characterization of instances that have a solution, our main technical contribution is a non-trivial reduction to 2-SAT.
This approach to showing polynomial running times was used before \cite{BMW18,GW09}, but we believe that its potential is still relatively unexplored.  

\section{Preliminaries and Preprocessing}\label{chapter_prelim}
We use graph-theoretical notation in a similar way as Diestel \cite{DBLP:books/daglib/diestel}.
For a graph~$G := (V,E)$ and a set of vertices~$W \subseteq V$, we use~$G[W]$ to denote the induced subgraph of~$W$ in~$G$, that is, the graph~$G' := (W, E')$ where~$\{v,w\} \in E'$ if and only if~$\{v,w\} \in E$ and~$v,w \in W$.
For two integers~$a$ and~$b$ we denote by~$[a,b]$ the set of integers~$\{a, a+1, ..., b\}$.
If~$G$ is a cycle, then we always assume that the agents are numbered from~$0$ to~$n-1$, where~$n = |V|$, such that agent~$i$ shares an edge with agents~$i+1 \bmod n$ and~$i-1 \bmod n$.
We also say that agent~$i+1 \bmod n$ is the clockwise neighbor of agent~$i$ and that agent~$i-1\bmod n$ is the counter-clockwise neighbor of agent~$i$.
We denote by~$\cycseq[a][b]$ the sequence of clockwise neighbors starting from~$a$ and ending in~$b$, that is, 
\begin{equation}
    \cycseq[a][b] =
	    \begin{cases*}
      (a,a+1,...,n-1,0,1,...,b) & if $b < a$ \\
      (a,a+1,...,b)        & otherwise
    \end{cases*}
\end{equation}

\begin{example}
Let~$n = 8$. Then~$\cycseq[2][6] = (2, 3, 4, 5, 6)$ and~$\cycseq[6][2] = (6, 7, 0, 1, 2)$.
\end{example}
We now formally define \ra{} and \ro{}.
Let~$N$ be a set of~$n$ agents and let~$X$ be a set of~$n$ indivisible objects.
Each agent~$i \in N$ has a \emph{preference list~$\succ_i$} over a non-empty subset~$X_i$ of the set of objects~$X$.
Each preference list is a strict ordering on~$X_i$.
The set of all preference lists is called a \emph{preference profile~$\succ$}.
A bijection~$\sigma : N \rightarrow X$ is called an\emph{ allocation or assignment}.
Akin to the house marketing problem \cite{AbrCecManMeh2005}, each agent is initially assigned exactly one object.
We denote this initial assignment by~$\sigma_0$ and for the sake of simplicity, we will interchangeably use~$\sigma(i)=o$ and~$(i,o) \in \sigma$ for an agent~$i$, an object~$o$, and an assignment~$\sigma$. 
Let~$G=(N,E)$ be a graph where the set of vertices is the set~$N$ of agents.
We will use the term agents interchangeably with vertices of~$G$.
A trade between agents~$i$ and~$j$ is only possible if their corresponding vertices share an edge in~$G$ and if both~$i$ and~$j$ receive an object that they prefer over the objects they hold in the current assignment~$\sigma$.
We can express this formally by~$\sigma(i) \succ_j \sigma(j)$ and~$\sigma(j) \succ_i \sigma(i)$ and we call such a trade a \emph{rational swap}.
A \emph{sequence of rational swaps} is a sequence of assignments~$(\sigma_s, ..., \sigma_t)$ where~$\sigma_i$ is the result of performing one rational swap in assignment~$\sigma_{i-1}$ for all~$i \in \{s + 1, s + 2, ..., t\}$.
We call an assignment~$\sigma$ \emph{reachable} if there exists a sequence of rational swaps~$(\sigma_0, ..., \sigma)$.
For an agent~$i$, an object~$x$ is reachable if there exists an assignment~$\sigma_t$ and a sequence of rational swaps~$(\sigma_0, ..., \sigma_t)$ such that~$\sigma_t(i) = x$.
Gourves et al. \cite{gourves_object_2017} introduced the problems \ra{} and \ro{} as follows.

\vspace{.3cm}
\problemdef{\ra}{A set~$N$ of agents, a set~$X$ of objects, a preference profile~$\succ$, a graph~$G$, an initial assignment~$\sigma_0$ and a target assignment~$\sigma$.}{Is~$\sigma$ reachable from~$\sigma_0$?}
\vspace{.1cm}
\problemdef{\ro}{A set~$N$ of agents, a set~$X$ of objects, a preference profile~$\succ$, a graph~$G$, an initial assignment~$\sigma_0$, an agent~$i$ and an object~$x$.}{Is~$x$ reachable for agent~$i$ from~$\sigma_0$?}
\vspace{.3cm}

Recall that assignments are bijections.
This allows us for any assignment~$\sigma$, any object~$y$ and the agent~$i$ with~$\sigma(i) = y$ to denote~$i$ by~$\sigma^{-1}(y)$.

We end this section with a simple data reduction rule for \ra{}. 
Whenever an agent prefers an object~$q$ over the object~$p$ it is assigned in the target assignment~$\sigma$, then we can simply remove~$q$ from its preference lists.
This is due to the fact that once the agent possesses~$q$, then it cannot receive~$p$ anymore.
We will assume that every instance has already been preprocessed by the following data reduction rule, which is equivalent to assuming that~$\sigma$ assigns each agent its most preferred object.
\begin{rrule}\label{rr1}
If for any agent~$i$ with object~$p := \sigma(i)$ there exists an object~$q$ such that~${q \succ_i p}$, then remove~$q$ from~$\succ_i$.
\end{rrule}

\section{A Polynomial-Time Algorithm for \ra{} on Cycles}\label{sec:cycles}
In this section we develop a polynomial-time algorithm for \ra{} on cycles.
To the best of our knowledge, this is the first polynomial-time algorithm for \ra{} beyond the initial algorithm for trees by Gourves et al.\,\cite{gourves_object_2017}.
Our algorithm generalizes several ideas used in the algorithm for trees and uses a novel characterization of solutions.
We divide the section into three subsections.
In Subsection \ref{subsec:directions}, we formally define what we mean by swapping an object in a certain direction and provide a polynomial-time algorithm to verify whether a given assignment of directions to all objects corresponds to a solution.
In this case, we will say that the assignment of directions \emph{yields~$\sigma$}.
In Subsection \ref{subsec:validity}, we define a property we call \emph{validity} and show that this characterizes the assignments of directions that yield~$\sigma$.
Finally in Subsection \ref{subsec:2SAT}, we reduce the problem of deciding whether there exists a valid assignment of directions to 2-SAT.

\subsection{Swapping Directions in a Cycle}
\label{subsec:directions}
In this subsection we will formally introduce assignments of directions.
We will refer to them as \emph{selections} and always denote them by~$\gamma$.
Consider an instance~$\mathcal{I} := (N, X, \succ, G, \sigma_0, \sigma)$ of \ra{}.
Let~$p$ be an object, let~$j = \sigma^{-1}(p)$, and let~$i = \sigma_0^{-1}(p)$.
Since the underlying graph is a cycle, there are exactly two paths between agents~$i$ and~$j$ for object~$p$.
By definition of rational swaps, once~$p$ has been swapped, say from agent~$i$ to agent~$i+1 \bmod n$ then~$p$ is not able to return to agent~$i$ since agent~$i$ just received an object that it prefers over~$p$ and will therefore not accept~$p$ again.
Hence, if~$p$ is swapped again, then it is given to agent~$i + 2 \bmod n$ and the argument can be repeated for agent~$i + 1 \bmod n$.
As there are only two paths between agents~$i$ and~$j$, there are also only two directions, namely clockwise and counter-clockwise.
We will henceforth encode these directions into a binary number saying that the direction of~$p$ is~$1$ if~$p$ is swapped in clockwise direction and~$0$ otherwise.
This yields the following definition of selections.
\begin{definition}
Let~$\mathcal{I} := (N, X, \succ, C_n, \sigma_0, \sigma)$ be an instance of \ra{}.
A \emph{selection}~$\gamma$ of~$\mathcal{I}$ is a function that assigns each object~$p \in X$ a direction~$\gamma(p) \in \{0,1\}$.
\end{definition}
Given a selection~$\gamma$, we say that an object~$p$ is \emph{closer} than another object~$q$ with~$\gamma(p) = \gamma(q)$ to~$\ell$, if starting from~$q$ and going in direction~$\gamma(q)$, it holds that~$p$ comes before~$\ell$.
Therein, $\ell$ can be an edge, an agent, or a third object.

We will see that for a pair of objects (with assigned directions) there is a unique edge over which they can be swapped.
To show this, we need the following definition of \emph{a path of an object}.
These paths consist of all agents that will hold the respective object in any successful sequence of swaps that respects the selection.

\begin{definition}\label{def:swapspace}
Let~$\gamma$ be a selection for instance~$\mathcal{I}$ and let~$p$ and~$q$ be two objects with~$\gamma(p) \neq \gamma(q)$.
Let~$i := \sigma_0^{-1}(p)$, let~$j = \sigma^{-1}(p)$, let~$I := \cycseq[i][j]$, and let~$J := \cycseq[j][i]$.
The \emph{path of~$p$ for~$\gamma$} is
\begin{equation*}
   P_\gamma(p) := \begin{cases*}
        C_n[I] & if $\gamma(p) = 1$ \\
		C_n[J]  & otherwise.
    \end{cases*}
\end{equation*}
The set~\emph{$\xi_\gamma(p,q)$ of shared paths of~$p$ and~$q$ for~$\gamma$} is the set of all connected paths in~$P_\gamma(p) \cap P_\gamma(q)$.\footnote{If~$\{\sigma_0^{-1}(p),\sigma^{-1}(p)\} \in E \cap P_{\gamma}(q)$, then~$\xi_\gamma(p,q)$ contains two disjoint paths, one with~$\sigma_0^{-1}(p)$ as endpoint and one with~$\sigma^{-1}(p)$.}
Finally,~$p$ and~$q$ are \emph{opposite} if there exists a selection~$\gamma'$ such that~$|\xi_{\gamma'}(p,q)| > 1$.
\end{definition}
Examples for paths and shared paths are given in \cref{fig:sharedrestpath}.
\begin{figure}
\centering
\begin{tikzpicture}
\def \radius {1.5};
\racycle{0}{0}{\radius}
\objpath{0}{0}{\radius}{225}{135}{p}{red}
\objpath{0}{0}{\radius}{-45}{45}{q}{blue}

\racycle{4.75}{0}{\radius}
\objpath{4.75}{0}{\radius}{225}{45}{p}{red}
\objpath{4.75}{0}{\radius}{-45}{135}{q}{blue}

\racycle{9.5}{0}{\radius}
\objpath{9.5}{0}{\radius}{-135}{135}{p}{red}
\objpath{9.5}{0}{\radius}{45}{315}{q}{blue}
\end{tikzpicture}
\caption{Given a selection~$\gamma$ and two objects~$p$ and~$q$ with~$\gamma(q) \neq \gamma(p) = 1$, the two marked paths in the left figure are the respective paths of~$p$ and~$q$ for~$\gamma$.
Since they do not intersect, the set~$\xi_\gamma(p,q)$ of shared paths of~$p$ and~$q$ for~$\gamma$ is empty.
In the center figure~$\xi_\gamma(p,q)$ contains the single intersection between~$\sigma^{-1}(q)$ and~$\sigma^{-1}(p)$ and in the right figure it contains the two intersections in the top and in the bottom.
The two objects are opposite in the left and the right figure, but not in the center.}%
\label{fig:sharedrestpath}
\end{figure}
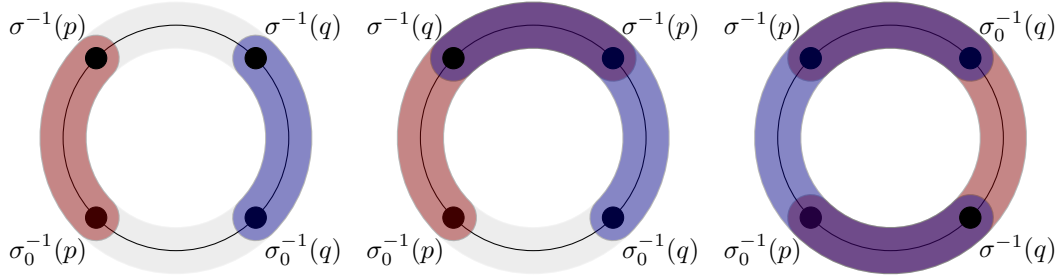
With these definitions at hand, we are now able to determine unique edges where two objects~$p$ and~$q$ can be swapped with respect to all selections~$\gamma$ with~$\gamma(q) \neq \gamma(p) = 1$.
\begin{lemma}\label{lem:uniqueSwap}
Let~$\gamma$ be a selection for instance~$\mathcal{I}$ and~$p$ and~$q$ be two objects with~$\gamma(p) \neq \gamma(q)$.
If for some path~$P \in \xi_\gamma(p,q)$ there is not exactly one edge in~$P$ such that the corresponding preference lists allow~$p$ and~$q$ to be swapped between the two incident agents, then~$\gamma$ does not yield target assignment~$\sigma$.
\end{lemma}
\begin{proof}
Note that~$p$ and~$q$ have to be swapped somewhere on each path in~$\xi_\gamma(p,q)$ as otherwise at least one of the objects cannot reach its target agent.
For each path~$P \in \xi_\gamma(p,q)$ there has to be at least one edge on~$P$ where~$p$ and~$q$ can be swapped or~$\gamma$ does not yield the target assignment~$\sigma$.
Assume towards a contradiction that there are at least two edges~$e$ and~$f$ on some path~$P \in \xi_\gamma(p,q)$ where~$p$ and~$q$ can be swapped according to the corresponding preference lists and~$\gamma$ yields~$\sigma$.
Suppose~$p$ and~$q$ are swapped over~$e$ in a sequence of rational swaps that yields~$\sigma$ (the case for~$f$ is analogous).
Then one of~$p$ or~$q$ must have already passed the other edge~$f$, say~$p$ (again, the case for~$q$ is analogous).
Since~$p$ and~$q$ could be swapped over~$f$, it holds for the two agents~$a$ and~$b$ incident to~$f$ that~$(p \succ_a q)$ and~$(q \succ_b p)$.
Since~$p$ already passed~$f$, agent~$a$ already held~$p$.
Hence, agent~$a$ will not accept object~$q$ in the future and hence~$q$ can not reach its destination as~$a$ is on the path of~$q$.
This contradicts the assumption that~$\gamma$ yields~$\sigma$.
\end{proof}

We show that two objects~$p$ and~$q$ are swapped exactly once on each path in~$\xi_\gamma(p,q)$ in any sequence of rational swaps and that~$|\xi_\gamma(p,q)| \leq 2$.
Examples with~$|\xi_{\gamma}(p,q)| \in \{0,1,2\}$ are given in \cref{fig:sharedrestpath}.

\begin{proposition}
\label{max2edge}
Let~$\gamma$ be a selection and let~$p$ and~$q$ be two objects with~$\gamma(q) \neq \gamma(p)$.
There are at most two edges on~$P_\gamma(p) \cup P_\gamma(q)$ such that~$p$ and~$q$ can only be swapped over these edges or~$\gamma$ does not yield target assignment~$\sigma$.
Each of these edges is on a different path in~$\xi_\gamma(p,q)$.
\end{proposition}

{
	\begin{proof}
	Let~$\gamma$ be any selection such that~$\gamma(q) \neq \gamma(p)$.
	Note that~$p$ and~$q$ cannot be swapped over any edge that is not on a path in~$\xi_\gamma(p,q)$ and that by \cref{lem:uniqueSwap} they have to be swapped over a specific edge for each path in~$\xi_\gamma(p,q)$.
	Hence the number of edges where~$p$ and~$q$ can be swapped over in any sequence of rational swaps that yield~$\sigma$ and that respects~$\gamma$ is equal to~$|\xi_\gamma(p,q)|$.
	
	Now assume that~$|\xi_{\gamma}(p,q)| \geq 3$.	
	Then,~$p$ and~$q$ are swapped at least thrice.	
	For~$p$ and~$q$ to perform a first swap, both of them must have at least passed one agent each.
	Observe that if~$p$ and~$q$ are swapped for a second time, then each agent has held~$p$ or~$q$ between the first and the second swap of~$p$ and~$q$.
	Hence, after the second swap, they must have passed at least~$n + 2$ agents combined.
	Repeating this argument once again, we get that they must have passed passed~$2n+2$ agents after the third swap.
	This means that at least one of the agents has passed more than~$n$ agents, a contradiction to the fact that an agent does not accept an object once the agent traded that object away.
	Thus, two objects can only be swapped twice and we can use \cref{lem:uniqueSwap} to find at most two unique edges.
	\end{proof}
}

Observe that $\xi_\gamma(p,q)$ only depends on $\gamma(p)$ and $\gamma(q)$ and hence for each selection~$\gamma'$ with~$\gamma'(p) \neq \gamma'(q) = \gamma(q)$, it holds that~$\xi_{\gamma}(p,q) = \xi_{\gamma'}(p,q)$ and the edges specified in \cref{lem:uniqueSwap} are the same for~$\gamma$ and~$\gamma'$.
We will denote the set of edges specified in \cref{lem:uniqueSwap} by~$E_\gamma(p,q)$.
Note that \cref{max2edge} also implies that if~$\gamma$ yields~$\sigma$, then~$|E_\gamma(p,q)| = |\xi_\gamma(p,q)| \leq 2$.
If~$\gamma(p) = \gamma(q)$, then we define~$E_\gamma(p,q) := \emptyset$.

We next show that the order in which objects are swapped is irrelevant once a selection is fixed.
We use the following definition to describe an algorithm that checks in polynomial time whether a selection yields~$\sigma$. 
\begin{definition}\label{def:swappos}
Let~$\mathcal{I}:=(N, X, \succ, C_n, \sigma_0, \sigma)$ be an instance of \ra{} and let~$\gamma$ be selection of~$\mathcal{I}$.
Let~$p$ and~$q$ be two objects with~$\gamma(q) \neq \gamma(p) = 1$.
Let~$i$ be the agent currently holding~$p$.
If~$q$ is held by agent~$i+1\bmod n$, then \emph{$p$ and~$q$ are facing each other} and if also~$\sigma(i) \neq p$ and~$\sigma(j) \neq q$, then \emph{$p$ and~$q$ are in swap position}.
\end{definition}

Using the notion of swap positions, we are finally able to describe a polynomial-time algorithm that decides whether a given selection yields~$\sigma$.
We refer to our algorithm as \texttt{Greedy Swap} and pseudo code for it is given in \cref{alg1}.
We mention that it is a generalization of the polynomial-time algorithm for \ra{} on trees by Gourves et al.\,\cite{gourves_object_2017}.
\texttt{Greedy Swap} arbitrarily swaps any pair of objects that is in swap position until no such pair is left.
If~$\sigma$ is reached in the end, then it returns True and otherwise it returns False.

\begin{proposition}
\label{greedySwaps}
Let~$\mathcal{I}:=(N, X, \succ, C_n, \sigma_0, \sigma)$ be an instance of \ra{} and let~$\gamma$ be a selection of~$\mathcal{I}$.
\texttt{Greedy Swap} returns True if and only if~$\gamma$ yields~$\sigma$.
\end{proposition}

\begin{algorithm}[t]
\KwData{\ra{} instance~$\mathcal{I} := (N, X, \succ, C_n, \sigma_0, \sigma)$ and selection~$\gamma$}

\If{$\sigma_0 = \sigma$}{\Return True\;}

\If{$\exists d \in \{0,1\}.\ \forall o \in X. \gamma(o) = d$}{\Return False\;}

$\sigma' \gets \sigma_0$\;
\While{$\exists x_1, x_2 \in X.\ x_1 \text{ and } x_2 \text{ are in swap position}$}{
  ~$i\gets \sigma'^{-1}(x_1)$\;
  ~$j\gets \sigma'^{-1}(x_2)$\;
  
   \eIf{$x_1 \succ_j x_2$ and~$x_2 \succ_i x_1$}{
   	 Swap~$x_1$ and~$x_2$\;
   	 Update~$\sigma'$\;
   }{
     \Return False\;
   }
 }
 
 \Return True\;
 \caption{Greedy Swap}
 \label{alg1}
\end{algorithm}

\begin{proof}
Observe that \texttt{Greedy Swap} only performs rational swaps and only returns True if~$\sigma$ is reached.
Hence, if it returns True, then~$\gamma$ yields~$\sigma$.

Now suppose that a selection~$\gamma$ yields assignment~$\sigma$, but \texttt{Greedy Swap} returns False.
Then either~$\sigma_0 \neq \sigma$ and every object is assigned the same direction, or there are two objects~$p$ and~$q$, which are in swap position at some edge~$e$, but the corresponding preference lists do not allow a swap.
In the former case~$\gamma$ clearly does not yield~$\sigma$.
In the latter case, consider the initial positions of~$p$ and~$q$ and the corresponding shared paths~$\xi_\gamma(p,q)$.
Since~$p$ and~$q$ are in swap position at the point where \texttt{Greedy Swap} returns False, it holds that~$\gamma(p) \neq \gamma(q)$ and there is a shared path~$P \in |\xi_{\gamma}(p,q)|$ such that~$e$ is on~$P$.

If~$p$ and~$q$ meet at edge~$e$ for the first time, then let~$m$ be the number of objects that are closer to~$q$ than~$p$ and that are assigned direction~$\gamma(p)$.
If~$p$ and~$q$ meet at edge~$e$ for the second time, then let~$m$ be the number of objects that are assigned direction~$\gamma(p)$.
We will now show that~$q$ has to be swapped with exactly~$m - 1$ objects before~$q$ and~$p$ can be swapped (for the first or between the first and the second swap, respectively).
We will only show the case for~$p$ and~$q$ being swapped over~$e$ for the first time as the other case is analogous when considering the instance after the first swap of~$p$ and~$q$ has just happened.
Starting from the agent that initially holds~$q$ we can then calculate the edge~$e'$ where~$q$ and~$p$ meet in every sequence of rational swaps that results in reaching the target assignment~$\sigma$.
Suppose~$q$ can be swapped with more than~$m-1$ objects before meeting~$p$ at edge~$e'$.
Then~$q$ must either be swapped with at least one object~$q'$ that is not closer to~$q$ than~$p$ or it must be swapped with the same object~$q'$ twice before being swapped with~$p$ for the first time.
This, however, means that~$q$ will meet~$q'$ after~$p$, a contradiction.
Now suppose that~$q$ can be swapped with less than~$m-1$ objects before meeting~$p$. Then there is an object~$r$ that is closer to~$q$ than~$p$ and that is assigned direction~$\gamma(p)$ and which is not swapped with~$q$ before~$q$ and~$p$ are swapped.
However, then~$q$ will not be able to meet~$p$ as at least~$r$ is between them, a contradiction.

Now that we have shown that~$q$ will swap exactly~$m$ objects before meeting~$p$, we can also determine the edge~$f$ where, given~$\gamma$, they must necessarily meet, in every sequence of rational swaps where objects are swapped according to the directions assigned to them by~$\gamma$.
Now since we assumed that~$\sigma$ is reachable with selection~$\gamma$, there must be a sequence of rational swaps where~$q$ and~$p$ must be swapped at edge~$f$. But since~$p$ and~$q$ met, by assumption, at edge~$e$, either~$e = f$, which is a contradiction because at~$e$ the preference lists of the incident agents do not allow a swap between~$p$ and~$q$, or~$p$ and~$q$ cannot have met at edge~$e$ if all swaps were performed according to the directions assigned by selection~$\gamma$, a contradiction to the definition of \texttt{Greedy Swaps}.
\end{proof}

\subsection{Validity of Selections}\label{subsec:validity}
We will next define a set of properties that a \emph{selection} needs to fulfill in order for \texttt{Greedy Swap} to return True and then show that these properties are both necessary and sufficient.
The idea is to have one property that guarantees that for each object and for each edge on its path there exists exactly one object that should be swapped with that object over this edge and a second property to guarantee that for each two objects~$p$ and~$q$ there are exactly~$|\xi_\gamma(p,q)|$ edges where~$p$ and~$q$ should be swapped.
Before we can formally state these properties, we give the following two definitions.
\begin{definition}
Let~$\mathcal{I}:=(N, X, \succ, C_n, \sigma_0, \sigma)$ be an instance of \ra.
The set~${H := X \times E(C_n)}$ is the set of all \emph{object-edge-pairs} and~$d: H \rightarrow \{0,1\}$ is the function that assign to each object-edge-pair~$(p,e) \in H$ the direction~$c$ such that~$e$ lies on the path of~$p$ from its initial agent to its destination in direction~$c$.
\end{definition}

Using the notion of object-edge-pairs, we define \emph{candidate lists}.
These will be used to eliminate possible choices for selections.
\begin{definition}\label{definition:validity:f_gamma}
Let~$(p,e) \in H$ be an object-edge-pair and let~$\gamma$ be a selection.
The \emph{candidate list}~$C(p,e)$ of~$p$ at edge~$e$ and the \emph{size}~$f_\gamma(p,e)$ of~$C(p,e)$ with respect to~$\gamma$ is
\begin{equation*}
    C(p,e) := \{q \in X \mid e \in E_\gamma(p,q)\} \text{ and }
\end{equation*}
\begin{equation*}
    f_\gamma(p,e) := 
    \begin{cases}
        |\{q \in C(p,e) \mid \gamma(p) \neq \gamma(q)\}| & \text{ if~$d(p,e) = \gamma(p)$} \\
        1& \text{ otherwise.}
    \end{cases}
\end{equation*}
\end{definition}
\Cref{fig:candidatelist} presents an example of candidate lists.
\begin{figure}
\centering
\begin{minipage}[c]{8cm}
\begin{tikzpicture}
    \node[circle, draw, thick] (1) at (0,0) {$1$};
    \node[circle, draw, thick] (2) at (2,0) {$2$};
    \node[circle, draw, thick] (3) at (4,0) {$3$};
    \node[circle, draw, thick] (4) at (6,0) {$4$};    
   
    \path[draw,thick]
    (1) edge node [midway,fill=white] {$e$} (2)
    (2) edge node [midway,fill=white] {$f$} (3)
    (3) edge node [midway,fill=white] {$g$} (4)
    (4) edge[bend left=25] node [midway,fill=white] {$h$} (1);
\end{tikzpicture}
\end{minipage}
\begin{minipage}[c]{4cm}
1:~$x_3 \succ x_2 \succ$ \fbox{$x_1$}\\
2:~$x_1 \succ x_3 \succ x_4 \succ$ \fbox{$x_2$}\\
3:~$x_4 \succ x_1 \succ$ \fbox{$x_3$}\\
4:~$x_2 \succ x_3 \succ$ \fbox{$x_4$}
\end{minipage}
\caption{Example for \ra{} on a~$C_4$ with edges~$e,f,g,h$.
The candidate lists of~$x_1$ are~$C(x_1, e) := \{ x_2, x_3\},\ C(x_1,h) := \emptyset,\ C(x_1, g) := \emptyset$, and~$C(x_1,f) := \{x_4\}$.
Note that since~$\sigma^{-1}(x_1) = 2$, the candidate list~$C(x_1, e)$ is defined such that~$x_1$ is swapped in clockwise direction while the candidate list of~$x_1$ for all other edges is defined for counter-clockwise direction.}
\label{fig:candidatelist}
\end{figure}
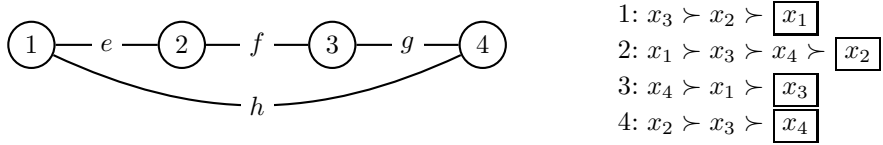
Note that the size of~$C(p,e)$ with respect to~$\gamma$ is set to one if~$d(p,e) \neq \gamma(p)$.
This is due to the fact that we will later search for a selection~$\gamma^*$ such that~$ f_{\gamma^*}(p,e)=1$ for all objects~$p$ and all edges with~$d(p,e) = \gamma^*(p)$.
This definition then avoids a case distinction.
The following observation follows from \cref{max2edge} and the observation that~$E_\gamma(p,q) = E_{\gamma'}(p,q)$ for all~$\gamma'$ with~$\gamma(p)=\gamma'(p)$ and~$\gamma(q)=\gamma'(q)$ with a simple counting argument.
\begin{observation}\label{observation:num.q.inCL}
	For each~$p \in X$ it holds that~$\sum_{e \in E} |C(p,e)| \leq 4|X|$.
\end{observation}

We can finally define the first property called \emph{exact} from the set of properties that characterizes the selections for which \texttt{Greedy Swap} returns True and which, by \cref{greedySwaps}, correspond to a solutions for \ra.
\begin{definition}\label{def:exact}
Let~$\mathcal{I} := (N, X, \succ, C_n, \sigma_0, \sigma)$ be an instance of \ra{}.
A selection~$\gamma$ is \emph{unambiguous} if for all~$(p,e) \in H$ it holds that~$f_\gamma(p,e) \leq 1$;~$\gamma$ is \emph{complete} if for all~$(p,e) \in H$ it holds that~$f_\gamma(p,e) \geq 1$, and \emph{exact} if~$f_\gamma(p,e) = 1$.
\end{definition}

Where exactness guarantees that for each object and for each edge on its path for~$\gamma$ there exists a possible swapping partner for this edge, the second property will guarantee that no pair of objects blocks one another in a selection.
We start with the case in which an object blocks another object that is swapped in the same direction.
\begin{definition}\label{def:shield}
Let~$\gamma$ be a selection and let~$p$ and~$q$ be two objects with~$\gamma(p) = \gamma(q)$.
Then~$q$ \emph{shields}~$p$ in direction~$\gamma(p)$ if~$P_\gamma(p) \cap P_\gamma(q) \neq \emptyset$, the object~$q$ is closer to~$p$ than the destination of~$p$ (in direction~$\gamma(p)$), and there exists an agent~$i$ on~$P_\gamma(p) \cap P_\gamma(q)$ with~$q \succ_i p$.
\end{definition}
In \cref{fig:candidatelist} the object~$x_1$ shields~$x_4$ in clockwise direction.
Note that if~$q$ shields~$p$ in direction~$c$, then~$p$ and~$q$ cannot be both swapped in direction~$c$ as the agent~$i$ obtains~$q$ first and then will not accept~$p$ afterwards.
However, by definition~$i$ is on the path of~$p$ and hence has to obtain~$p$ or the target assignment~$\sigma$ cannot be reached.

We continue with the case where an object blocks another object that is swapped in the opposite direction.
Based on \cref{lem:uniqueSwap}, we define compatibility of two objects.
\begin{definition}\label{def:compatible}
Let~$\gamma$ be a selection and let~$p$ and~$q$ be two objects with~$\gamma(p) \neq \gamma(q)$.
Then,~$p$ and~$q$ are \emph{compatible} in selection~$\gamma$ if for each~$P \in \xi_\gamma(p,q)$ there exists exactly one edge on~$P$ that is also contained in~$E_\gamma(p,q)$.
\end{definition}

Based on \cref{def:shield,def:compatible} we formalize the second property that characterize valid selections as follows.
\begin{definition}\label{def:harmonic}
A selection~$\gamma$ is \emph{harmonic} if for every object~$p$ there is no object~$q$ moving in direction~$\gamma(p)$ that shields~$p$ in that direction,
and every object~$r$ with direction~$1-\gamma(p)$ is compatible with~$p$.
A selection is \emph{valid}, if it is both exact and harmonic.
\end{definition}

In order to show that valid selections characterize those selections that correspond to a solution of \ra{}, we start with an intermediate lemma.

\begin{lemma}\label{order_lemma}
Let~$\mathcal{I} := (N, X, \succ, C_n, \sigma_0, \sigma)$ be an instance of \ra{} on cycle~$C_n$, let~$\gamma$ be a selection, let~$(p,e_0),(p,e_1) \in H$ be two object-edge-pairs where~$d(p,e_0) = d(p,e_1)$ and let~$q_0, q_1$ be objects with direction~$1 - \gamma(p)$. Let further~$q_0 \in C(p,e_0)$ and let~$q_1 \in C(p,e_1)$. Let~$h \in \{0,1\}$ be the index such that~$q_h$ starts closer to~$p$ in direction~$\gamma(p)$ than~$q_{1-h}$. 
If~$\gamma$ is harmonic, then~$e_h$ is closer to~$p$ in direction~$\gamma(p)$ than~$e_{1-h}$. If~$e_{1-h}$ is closer to~$p$ in direction~$\gamma(p)$ than~$e_{h}$, then~$q_{h}$ shields~$q_{1-h}$.
\end{lemma}
\begin{proof}
Considering the preference lists at edges~$e_{1-h}$ and~$e_h$ we will denote the agents at~$e_{1-h}$ by~$k$ and~$k+1$ and the agents at~$e_h$ by~$m$ and~$m+1$.\\
Suppose towards a contradiction that~$e_{1-h}$ is closer to~$p$ than~$e_h$ in direction~$d(p,e_h)$.

We distinguish between two cases. In the first case~$\gamma$ is harmonic. Then since~$q_h$ and~$q_{1-h}$ are assigned the same direction,~$q_h$ and~$q_{1-h}$ are not opposite. Further~$q_h$ and~$q_{1-h}$ cannot shield each other. Thus, it must hold that
\begin{equation*}
q_{1-h} \succ_{k+1} q_h,
\end{equation*}
because~$k+1$ is on both~$P_\gamma(q_h)$ and~$P_\gamma(q_{1-h})$.

In the second case~$\gamma$ is not harmonic and there is some agent that will not accept~$q_{1-h}$ since it prefers~$q_h$ more and~$q_{1-h}$ would be shielded from its destination by~$q_h$. Therefore we will now assume that~$\gamma$ is harmonic and show that then~$e_h$ must be closer to~$p$ in direction~$\gamma(p)$ than~$e_{1-h}$.

Since~$p$ can be swapped with~$q_{1-h}$ at~$e_{1-h}$ we know that 
\begin{equation*}
p \succ_{k+1} q_{1-h}
\end{equation*}
since~$q_{1-h}$ arrives at~$k+1$ and is then swapped with~$p$.
However~$e_h$ is further away from~$p$ than~$e_{1-h}$ and in order for~$q_h$ to be swapped with~$p$ at that edge we know
\begin{equation*}
q_h \succ_{m} p
\end{equation*}
We also know that~$p$ has already passed agent~$k+1$ to get to~$m$ and that~$q_h$ must at least come to agent~$k-1$ since otherwise~$q_{1-h}$ could not reach its destination. Since~$q_{1-h}$ can be swapped with~$p$ at~$e_{1-h}$, its destination is at least agent~$k$. But if~$q_h$ must move at least to~$k-1$ it also must pass~$k+1$ and since~$p$ was there first we know that 
\begin{equation*}
q_h \succ_{k+1} p
\end{equation*}
because otherwise~$p$ would shield~$q_h$ from its destination which we know cannot be true as~$\gamma$ is harmonic.
However, if those terms are true then
\begin{equation*}
q_h \succ_{k+1} p \succ_{k+1} q_{1-h} \succ_{k+1} q_h,
\end{equation*}
which cannot be true as~$\succ$ is a strict ordering, a contradiction.
\end{proof}

We end this subsection with the statement that valid selections are exactly the selections that correspond to a solution of \ra.

\begin{proposition}
\label{validDecides}
Let~$\mathcal{I} := (N, X, \succ, C_n, \sigma_0, \sigma)$ be an instance of \ra{} and let~$\gamma$ be a selection for~$\mathcal{I}$.
\texttt{Greedy Swap} returns True on input~$\gamma$ if and only if~$\gamma$ is valid.
\end{proposition}

\begin{proof}
First we will prove that if~$\gamma$ is valid then \texttt{Greedy Swap} returns True. 
Let therefore~$q$ be an object with path~$P := P_\gamma(q)$. We will prove that if~$\gamma$ is valid then~$q$ is guaranteed to pass every edge of~$P$. Since~$q$ is arbitrarily chosen we can generalize the statement to all objects. Now recall that \texttt{Greedy Swap} simply performs a swap between two objects that are in swap position over and over again. It only returns False if there are two objects in swap position at edge~$e$ that cannot be swapped according to the preference lists of the agents that are incident to~$e$. If we prove for every object, at every edge on its path, that the corresponding preference lists always allow a swap, we know that \texttt{Greedy Swap} returns True.

Consider any edge~$e_i$ on the path~$P_\gamma(q)$ of~$q$.
Let~$Q$ be the path starting with~$e_i$ and ending at~$\sigma^{-1}(q)$, that is, the path that~$q$ still has to take after reaching~$e_i$.
Let~$Q' = P_\gamma(q) \setminus Q$ be the path that~$q$ has already passed when reaching~$e_i$.
Let~$t_i$ be the object that~$q$ meets at~$e_i$ with~$\gamma(q) \neq \gamma(t_i)$.
Suppose~$q$ and~$t_i$ cannot be swapped at~$e_i$ due to the preference lists of their agents. Then either there exists no object that~$q$ can be swapped with at edge~$e_i$ in which case~$\gamma$ is incomplete, a contradiction, or there exists object~$q'$ with~$\gamma(q') \neq \gamma(q)$ that can be swapped with~$q$ at~$e_i$ but~$t_i$ starts closer to~$q$ than~$q'$. We will now show that if~$e_i$ is not an edge where~$q$ and~$t_i$ can be swapped, then that leads to a contradiction.

First of all, since~$\gamma$ is harmonic there must be an edge~$f$ on the shared path of~$q$ and~$t_i$ for~$\gamma$ where~$t_i$ and~$q$ can be swapped. We distinguish between two cases. If~$f$ is on~$Q'$, then~$q$ has already passed edge~$f$ and there exist two objects with direction~$1 - \gamma(q)$ that~$q$ can be swapped with at~$f$. Then,~$\gamma$ would be ambiguous, a contradiction. Otherwise~$f$ is on~$Q$ but not equal to~$e_i$. But then the edge where~$t_i$ and~$q$ can be swapped at is further away from~$q$ with respect to~$\gamma$, than the edge where~$q$ and~$q'$ can be swapped at, even though~$t_i$ is closer to~$q$ than~$q'$, with respect to~$\gamma$. But then due to \cref{order_lemma},~$\gamma$ is not harmonic, a contradiction.
Thus we can conclude that~$q$ and~$t_i$ can be swapped at edge~$e_i$.
Since we chose the edge~$e_i$ to be arbitrary we can generalize this to all edges on~$P_\gamma(q)$. From this follows that~$q$ reaches its destination. This can be generalized to all objects and thus, as described above, \texttt{Greedy Swap} returns True.

Now we will show that if~$\gamma$ is not valid, then \texttt{Greedy Swap} returns False.
Let~$\gamma$ be incomplete. Then there exists an object~$p$ and an edge~$e$ on its path where there exists no object~$q$ such that~$p$ and~$q$ can be swapped at~$e$.
Hence, if~$p$ reaches this edge there is no object that it can be swapped with.
Since this edge is on its path,~$p$ cannot have reached its destination and hence~$\sigma$ cannot be reached.
\texttt{Greedy Swap} therefore returns False.

Let~$\gamma$ be ambiguous. Then there exist objects~$p$,~$q$ and~$r$ and edge~$e$ on their paths such that if~$p$ reaches~$e$, then it can be swapped with both~$q$ and~$r$. Say~$p$ is swapped with~$q$ at~$e$.
Note that if~$p$ and~$r$ can be swapped at edge~$e$ then~$r$ must at least reach both agents incident to~$e$ to get to its destination and at one of these two agents, let us call it~$i$ it holds~$p \succ_i r$.
But then~$r$ cannot be accepted by~$i$ anymore since~$p$ has already been assigned to both of the agents before~$r$ has been assigned to one of them and thus~$r$ cannot reach its destination. Hence,~$\sigma$ cannot be reached. \texttt{Greedy Swap} therefore returns False.

Let~$\gamma$ be inharmonic. We distinguish between two cases.
In the first case there exist objects~$p$ and~$q$ such that~$q$ shields~$p$ from its destination. That means that~$p$ and~$q$ walk in the same direction and~$q$ stops at its destination~$i$ before~$p$ can pass~$i$. Since~$i$, however, is on~$P_\gamma(p)$,~$p$ cannot reach its destination and~$\sigma$ cannot be reached.
In the second case there exist objects~$p$ and~$q$ with~$\gamma(p) \neq \gamma(q)$ and a path~$P \in \xi_\gamma(p,q)$ such that there exists no edge on~$P$ where~$p$ and~$q$ can be swapped.
Since~$P \in \xi_\gamma(p,q)$, it holds that~$p$ and~$q$ must meet at some edge~$e$ on~$P$ and none of them has already reached its destination.
Since~$p$ and~$q$ cannot be swapped at~$e$ they cannot reach their destinations and \texttt{Greedy Swap} returns False.

If~$\gamma$ is not valid, then it is either incomplete, ambiguous or inharmonic and thus \texttt{Greedy Swap} returns False on input~$\gamma$, a contradiction.
\end{proof}

\subsection{Reduction to 2-SAT}\label{subsec:2SAT}
In this subsection we will construct a 2-SAT formula~$\phi$ such that a selection is valid if and only if it corresponds to a satisfying truth assignment of~$\phi$.
We will use a variable for each object and say that setting this variable to True corresponds to swapping this object in clockwise direction.
We will use objects interchangeably with their respective variables.
The formula~$\phi$ will be the conjunction of two subformulas~$\psi_h$ and~$\psi_e$.
The first subformula~$\psi_h$ corresponds to harmonic selections and is constructed as follows.
\begin{construction}
\label{const:harmonic}
Let~$\mathcal{I}:=(N, X, \succ, C_n, \sigma_0, \sigma)$ be an instance of \ra.
For every pair~$p,q$ of objects for which there exists a direction~$c$ such that~$q$ shields~$p$ in direction~$c$, we add the clause~$p \rightarrow \neg q$ if~$c = 1$, and~$\neg p \rightarrow q$ otherwise.
Further, for every pair~$p,q$ of objects for which there exists a direction~$c$ such that~$p$ and~$q$ are not compatible for any selection~$\gamma$ where~$c = \gamma(p) \neq \gamma(q)$, we add the term~$p \rightarrow q$ if~$c = 1$, and~$\neg p \rightarrow \neg q$ otherwise.
We use~$\psi_h$ to denote the constructed formula.
\end{construction}

We will now show that any satisfying truth assignment of~$\psi_h$ exactly corresponds to a harmonic selection.
\begin{lemma}
\label{harmonicSel}
Let~$\mathcal{I}:=(N, X, \succ, C_n, \sigma_0, \sigma)$ be an instance of \ra{}.
A selection~$\gamma$ is harmonic if and only if it corresponds to a satisfying truth assignment of~$\psi_h$.
\end{lemma}

\begin{proof}
We will prove both directions of the statement by contradiction. For the first direction suppose that~$\gamma$ is a harmonic selection but it does not correspond to a satisfying truth assignment of~$\psi_h$.
Then there must exist a clause~$c$ in~$\psi_h$ such that~$c$ is not satisfied by the truth assignment corresponding to~$\gamma$.
We distinguish between two cases.
In the first case~$c$ is a clause for~$q$ shielding~$p$ in direction~$c$ for every selection~$\gamma'$ where~$\gamma'(p) = c_p = \gamma'(q)$ but it holds that~$\gamma(p) = c_p = \gamma(q)$.
But then, by definition,~$\gamma$ is not harmonic, a contradiction.
In the second case~$c$ is a clause for two non-compatible objects~$p$ and~$q$ for every selection~$\gamma$ where~$\gamma(p) = c_p = 1-\gamma(q)$ but it holds that~$\gamma(p) = c_p = 1-\gamma(q)$, a contradiction to~$\gamma$ being harmonic.

We will now show the other direction of the statement.
Suppose that~$\gamma$ corresponds to a satisfying truth assignment of~$\psi_h$ but it~$\gamma$ is not harmonic.
We again distinguish between two cases.	
In the first case there exist two objects~$p$ and~$q$ and a direction~$c_p \in \{0,1\}$ such that~$q$ shields~$p$ in direction~$c_p$ for every selection~$\gamma'$ where~$\gamma'(p) = c_p = \gamma'(q)$. But since~$\gamma$ corresponds to a satisfying truth assignment of~$\psi_h$, due to the clauses for shielding objects, if~$\gamma(p) = c_p$, then~$\gamma(q) = 1-c_p$ and thus,~$p$ does not shield~$q$, a contradiction.

In the second case there exist two objects~$p$ and~$q$ and a direction~$c_p \in \{0,1\}$ such that~$p$ and~$q$ are not compatible in direction~$c_p$ for every selection~$\gamma'$ where~$\gamma'(p) = c_p = 1- \gamma'(q)$. But since~$\gamma$ corresponds to a satisfying truth assignment of~$\psi_h$, due to the clause for compatibility, if~$\gamma(p) = c_p$, then~$\gamma(q) = c_p$ and thus,~$p$ and~$q$ are compatible, a contradiction.
\end{proof}

A formula for exact selections requires much more work.
To this end, we first introduce a variant of \ra.
\vspace{.3cm}
\problemdef{\fsra}{An instance~$\mathcal{I} := (N, X, \succ, C_n, \sigma_0, \sigma)$ of \ra{} and an edge~$e \in E(C_n)$}{Is~$\sigma$ reachable if the first swap is performed over edge~$e$?}
\vspace{.3cm}
Observe that~$\mathcal{I}$ is a yes-instance of \ra{} if and only if there is an edge~$e$ in~$C_n$ such that~$(\mathcal{I},e)$ is a yes-instance of~\fsra.
We will later iterate over all edges and check for each whether it yields a solution.
Suppose that some instance~$(\mathcal{I},e_i)$ of \fsra{} is a yes-instance.
Let~$e_i = \{i,i+1 \bmod n\}$ and~$j = i+1 \bmod n$.
Then, the object~$\sigma_0(i)$ is swapped into clockwise direction and the object~$\sigma_0(j)$ is swapped into counter-clockwise direction.
Thus, in every selection~$\gamma$ that yields the target assignment~$\sigma$ as a solution to~$(\mathcal{I},e_i)$, it must hold that~$\gamma(\sigma_0(i)) = 1$ and~$\gamma(\sigma_0(j)) = 0$.
We refer to the pair of objects~$x:= \sigma_0(i)$ and~$y:= \sigma_0(j)$ as the \emph{guess} of that instance and denote it~$\Phi := (x,y)$, where~$x$ refers to the object that is swapped into clockwise direction and~$y$ refers to the object that is swapped into counter-clockwise direction.
For every selection~$\gamma$ where~$\gamma(x) = 1 = 1 - \gamma(y)$, we say that~$\gamma$ \emph{respects} the guess~$\Phi := (x,y)$.
Further,~$\Phi_0$ denotes the object in guess~$\Phi$ which is swapped in counter-clockwise direction and~$\Phi_1$ denotes the object in guess~$\Phi$ which is swapped in clockwise direction.
We also denote the unique path of~$x$ in clockwise direction from its initial agent to its destination with~$P_x$ and the unique path of~$y$ in counter-clockwise direction from its initial agent to its destination with~$P_y$.
If we can prove that for a guess~$\Phi$ and the corresponding instance of \fsra{} there exists no valid selection that respects~$\Phi$, then we say that the guess~$\Phi$ is \emph{wrong}.
We will now define a property of objects which cannot be swapped in one direction given a certain guess~$\Phi$.
\begin{definition}\label{def:decidedObject}
Let~$\mathcal{I} := ((N, X, \succ, C_n, \sigma_0, \sigma), e)$ be an instance of \fsra{} with guess~$\Phi := (x,y)$ and let~$p$ be an object such that there exists a direction~$d$ such that there is no selection~$\gamma$ with~$\gamma(p) = d$ that respects~$\mathcal{I}$ and that is valid.
Then~$p$ is \emph{decided} in direction~$1-d$.
\end{definition}
We find that an object~$p \notin \{x,y\}$ is decided by means of five rules listed below.
Therein,~$q$ is an object in~$\{x,y\}$.
\begin{enumerate}
	\item The objects~$p$ and~$q$ are opposite.
	\item The object~$p$ occurs in none of~$q$'s candidate lists.
	\item The object~$p$ occurs in a candidate list of~$q$ together with another object~$r$ that is decided in direction $1-d(p,e)$.
	\item The agent~$\sigma(p)$ is neither on~$P_x$ nor on~$P_y$.
	\item The object~$p$ starts between two decided and opposite objects.
\end{enumerate}
We mention that there are a couple of additional technical conditions for the last rule and that there may be decided objects that are not found by either of the five rules.
For the sake of simplicity, we say that an object is decided only if fulfills at least one of the rules above.
All other objects are said to be \emph{undecided}.
We next show that each of these rules in fact shows that~$p$ is decided.

\begin{lemma}\label{lemma:decidedOpposite}
Let~$\Phi := (x,y)$ be a guess, let~$c$ be the decided direction of~$x$ in~$\Phi$, and let~$p \in X \setminus \{x,y\}$.
If~$p$ and~$x$ are opposite, then~$p$ is decided in the direction~$1-c$.
\end{lemma}
\begin{proof}
	Observe that by definition there is a direction~$d$ such that in every selection~$\gamma$ with~$\gamma(q) \neq \gamma(p) = d$ it holds that~$|\xi_\gamma(p,q)| \geq 2$.
	If~$c=d$, then in any selection~$\gamma'$ with~$\gamma'(p) = \gamma(q) = c$, it holds that~$q$ shields~$p$.
	If~$c \neq d$, then in any selection~$\gamma'$ with~$\gamma'(p) = \gamma(q) = c$, it holds that~$p$ shields~$q$.
	Hence in no selection~$\gamma'$ with~$\gamma'(p) = \gamma'(q)$ is harmonic and thus there is no valid selection~$\gamma'$ with~$\gamma'(q) = c$ and thus~$q$ is decided in direction~$1-c$.
\end{proof}

The second rule focuses on objects that cannot be swapped with one of the two guessed objects.
\begin{lemma}\label{lemma:decidedObjects}
Let~$\Phi := (x,y)$ be a guess and let~$p \in X$ be an object such that~$y \neq p \neq x$.
If there exists a~$c \in \{0,1\}$ such that~$p$ is not in the candidate list of~$\Phi_c$, then~$p$ is decided.
\end{lemma}
\begin{proof}
Let~$\gamma$ be a selection where
\begin{equation*}
\gamma(p) \neq \gamma(\Phi_c).
\end{equation*}
We distinguish between two cases.
In the first case the set~$\xi_\gamma(p, \Phi_c)$ of shared paths is empty.
Then, the path~$P_\gamma(p)$ of~$p$ and~$P_\gamma(\Phi_c)$ of~$\Phi_c$ are disjoint.
Now consider a selection~$\gamma'$ where
\begin{equation*}
\gamma(p) = \gamma(\Phi_c).
\end{equation*}
Because~$P_{\gamma'}(p) \cup P_{\gamma}(p) = C_n$, we now know that
\begin{equation*}
    P_{\gamma'}(\Phi_c) \subset P_{\gamma'}(p)
\end{equation*}
and thus~$\Phi_c$ shield~$p$ from its destination in~$\gamma'$. Thus, a selection~$\gamma'$ where~$\gamma'(p) = \gamma'(\Phi_c)$ is not valid.

In the second case there is at least one path in~$\xi_\gamma(p, \Phi_c)$.
But then since~$p$ is not in the candidate list of~$\Phi_c$, there exists no edge on any path~$P \in \xi_\gamma(p, \Phi_c)$ where~$p$ and~$\Phi_c$ can be swapped and thus~$p$ and~$\Phi_c$ are not compatible in selection~$\gamma$.
Thus, a selection~$\gamma'$ where~$\gamma'(p) \neq \gamma'(\Phi_c)$ is not valid.
This concludes the proof.
\end{proof}

The third rule focuses on the effect that decided objects have on other objects.

\begin{lemma}
\label{lemma:candidateList:decidedByDecided}
Let~$(p,e) \in H$ such that~$p \in \Phi$, i.e,~$p$ is a guessed object. If there exists an object~$q \in C(p,e)$ such that~$q$ is decided in direction~$1-d(p,e)$, then for every~$q' \in C(p,e)$ where~$q' \neq q$ it holds that~$q'$ is decided in direction~$d(p,e)$.
\end{lemma}
\begin{proof}
Suppose towards a contradiction that there exists an undecided object~$q' \in C(p,e)$ and a valid selection~$\gamma$ that respects~$\Phi$ such that~$\gamma(q') = 1 -d(p,e)$.
Since~$q$ is decided in direction~$1-d(p,e)$, it holds for every valid selection~$\gamma'$ that~$\gamma'(q) = 1-d(p,e)$ and thus, also for~$\gamma$.
But then, by definition,~$f_\gamma(p,e) > 1$ and thus,~$\gamma$ is not valid, a contradiction.
\end{proof}

The fourth rule focuses on objects that are in the end held by the agents that are neither on~$P_x$ nor on~$P_y$.

\begin{lemma}\label{lemma:candidateLists:decidedByDestination}
Let~$\mathcal{I} := ((N, X, \succ, C_n, \sigma_0, \sigma), e')$ be an instance of \fsra{} on cycle~$C_n$ with guess~$\Phi := (x,y)$ and let~$x$ and~$y$ be non-opposite.
Let~$q$ be an object such that the destination~$\sigma(q)$ of~$q$ is not on the union of~$P_x$ and~$P_y$. Then, for every~$(p,e) \in H$ with~$p \in \{x,y\}$,~$q$ is decided in direction~$d(p,e)$.
\end{lemma}
\begin{proof}
Since~$P_x$ and~$P_y$ only intersect at the edge~$e'$ where~$x$ and~$y$ are initially swapped according to the instance~$\mathcal{I}$ of \fsra{} as defined above, there exists a direction~$c \in \{0,1\}$ such that~$q$ is not initially assigned to an agent on the path of guessed object~$\Phi_c \in \{x,y\}$.
But then, since, by assumption, the destination of~$q$ is also not on the path of~$\Phi_c$ for both~$c=1$ and~$c=0$, for any selection~$\gamma$ that respects~$\Phi$ and where~$\gamma(q) = c$,~$P_\gamma(\Phi_c) \subset P_\gamma(q)$.

However, if that is true then, by definition,~$\Phi_c$ shields~$q$ from its destination and thus, there exists no valid selection~$\gamma$ such that~$\gamma(q) = c$. Therefore,~$q$ is decided in direction~$1-c$. 
\end{proof}

Lastly, the fifth rule focuses on objects that start between two decided and opposite objects.
We only consider the case where one of the two objects is a guessed object and the two objects are not swapped.

\begin{lemma}\label{lemma:between_decided}
Let~$\Phi$ be a guess, let~$p \in \Phi$, let~$q$ be an object decided in the same direction~$d$ as~$p$ and let~$p'$ denote the other guessed object.
If~$\xi_\gamma(q,p') = \emptyset$ for all selections~$\gamma$ with~$\gamma(p') \neq \gamma(q) = d$, then each object~$r$ that starts on the path between~$p$ and~$q$ in direction~$d$ is decided.
After an~$O(n^3)$-time preprocessing, the set of these objects decided in direction~$d$ can be computed in linear time.
\end{lemma}
\begin{proof}
Suppose towards a contraction that $r$ is not decided.
By \cref{lemma:decidedObjects} it is in a candidate list of $p'$.
However, since $q$ starts closer to $p'$ in direction $d$ than $r$, the path of $r$ to its destination in direction $d$ must fully contain the path of $q$ to its destination in direction $d$, as the shared path between $q$ and $p'$ is empty.
But then $q$ shields $r$ from its destination and since $q$ is decided in direction $d$, there exists no valid selection $\gamma$ such that $\gamma(r) = d$ and thus, $r$ is decided in direction $1-d$, a contradiction.

One can precompute in~$O(n^3)$ time for each pair~$p,q$ of objects and each possible direction assignment for~$p$ and~$q$ whether~$p$ shields~$q$ for these direction assignments.
We first precompute for each agent~$a$ and each object~$p$ the position of~$p$ in~$a$'s preference list.
This allows us to compare two objects in the preference list of an agent in constant time.
Afterwards we iterate in~$O(n^3)$ time over all pairs~$p,q$ of objects, all possible direction assignments for~$p$ and~$q$, and all agents and check whether this agent prefers~$p$ or~$q$.
After the preprocessing, we check for each~$r$ that starts on the path between~$p$ and~$q$ in direction~$d$ in constant time whether it is shielded by~$q$ in direction~$d$.
If this is the case, then it is decided in direction~$1-d$ and if this is not the case, then it is shielded by~$p'$ in direction~$1-d$ and hence is decided in direction~$d$.
Since there are at most~$n$ such objects~$r$, all of these computations take~$O(n)$ time after the preprocessing.
\end{proof}

The next goal is to show that, after performing the five rules exhaustively, there are at most two undecided objects in each candidate list of the two guessed objects.
To this end, we proceed by partitioning all objects with respect to the candidate lists they appear in.
Afterwards, we presents a sequence of auxiliary lemmas leading up to the proposition that the number of undecided objects in a candidate list of a guessed object is zero or two.
The following definition introduces various sets of objects and candidate lists that can be defined per guess~$\Phi$.
	
\begin{definition}\label{definition:candidateLists:C_UandC_D}
Let~$\Phi := (x,y)$ be a guess.
We partition the set of objects~$X$ into three subsets~$U, D$ and~$D_0$ where~$U$ is the set of undecided objects,~$D$ is the set of decided objects that appear in at least one candidate list~$C(p,e)$ of a guessed object~$p \in \{x,y\}$ with direction~$1-d(p,e)$ and where~$D_0 = X \setminus (U \cup D)$.
Further, let~$O \subseteq D$ denote the set of objects~$q$ for which there exists a guessed object~$p \in \{x,y\}$ such that~$p$ and~$q$ are opposite and there exists an edge~$e$ such that~$q \in C(p,e)$ and~$q$ is decided in direction~$1-d(p,e)$.
Moreover, let~$\mathcal{C}$ be the set of candidate lists~$C(p,e)$ with~$p \in \{x,y\}$.
Lastly, let~$\mathcal{C}_D$ and~$\mathcal{C}_U$ be subsets of~$\mathcal{C}$ where~$\mathcal{C}_D$ contains all candidate lists which contain an object~$q \in D$ and where~$\mathcal{C}_U$ contains all candidate lists which contains at least two distinct objects~$q_0,q_1 \in U$.
\end{definition}
	
By definition,~$U, D$ and~$D_0$ are a partition of~$X$.	
\begin{observation}\label{obs:partition}
    Sets~$U, D$ and~$D_0$ are a partition of~$X$.
\end{observation}

We will further show that~$\mathcal{C}_D$ and~$\mathcal{C}_U$ are a partition of~$\mathcal{C}$.

\begin{lemma}\label{lemma:candidateLists:C_UandC_D}
Let~$\Phi := (x,y)$ be a guess.
Then~$\mathcal{C}_D$ and~$\mathcal{C}_U$
are a partition of~$\mathcal{C}$ or~$\Phi$ is wrong.
\end{lemma}
\begin{proof}
We prove the above statement in two steps. First we prove that~$\mathcal{C}_D \cap \mathcal{C}_U = \emptyset$. Afterwards we prove that~$\mathcal{C}_D \cup \mathcal{C}_U = \mathcal{C}$ or the target assignment~$\sigma$ cannot be reached with guess~$\Phi$. As a result,~$\mathcal{C}_D$ and~$\mathcal{C}_U$ are partition of~$\mathcal{C}$ or the target assignment~$\sigma$ cannot be reached with guess~$\Phi$.

According to \cref{lemma:candidateList:decidedByDecided}, a candidate list~$C(p,e)$ cannot contain two objects~$q_0$ and~$q_1$ such that~$q_0$ is decided in direction~$1-d(p,e)$ and~$q_1$ is undecided. Therefore~$\mathcal{C}_D \cap \mathcal{C}_U = \emptyset$.

We will now prove that~$\mathcal{C}_D \cup \mathcal{C}_U = \mathcal{C}$ or the target assignment~$\sigma$ cannot be reached with guess~$\Phi$. 
Suppose there exists a candidate list~$C(p,e)$ such that~$C(p,e) \not \in \mathcal{C}_D$ and~$C(p,e) \not \in \mathcal{C}_U$.
Then~$C(p,e)$ contains only objects that are decided in direction~$d(p,e)$.
However, that means in every valid selection~$\gamma$, for every object~$q$ in~$C(p,e)$ it holds that~$\gamma(q) = d(p,e)$. Further, since~$p$ is a guessed object,~$\gamma(p) = d(p,e)$.
Thus,~$f_\gamma(p,e) = 0$ and~$\gamma$ is not valid, a contradiction. Since consequentially there exists no valid selection,~$\sigma$ cannot be reached with guess~$\Phi$.

Thus,~$\mathcal{C}_D$ and~$\mathcal{C}_U$ are partition of~$\mathcal{C}$ or the target assignment~$\sigma$ cannot be reached with guess~$\Phi$.
\end{proof}

The following two lemmas are auxiliary lemmas for proving that the number of undecided objects in a candidate list of a guessed object is at most two. 
The next lemma shows that every undecided objects is in exactly two candidate lists in~$\mathcal{C}$, that is, in candidate lists of guessed objects.

\begin{lemma}\label{lemma:candidateLists:undecidedMax2}
Let~$\Phi$ be a guess. For every undecided object~$q$, there exist exactly two candidate lists~$C_0, C_1 \in \mathcal{C}$ such that~$q \in C_0$ and~$q \in C_1$ or the target assignment~$\sigma$ cannot be reached with guess~$\Phi$.
\end{lemma}
\begin{proof}
Suppose the statement is not true. Then either~$q$ appears only in one candidate list and~$q$ is decided due to \cref{lemma:decidedObjects} or
there exist three candidate lists~$C_0, C_1, C_2 \in \mathcal{C}$ such that~$q \in C_0, q \in C_1$ and~$q \in C_2$. But then there must exist a guessed object~$p \in \{x,y\}$ such that~$q$ is in two candidate lists of~$p$. We make the following case distinction.

In the first case, $q$ and~$p$ are opposite.
Then~$q$ is decided due to \cref{lemma:decidedOpposite}.
In the second case~$q$ and~$p$ can be swapped twice on the same shared path but then, due to \cref{lem:uniqueSwap}, the target assignment~$\sigma$ cannot be reached with guess~$\Phi$ and hence,~$\Phi$ is wrong.
\end{proof}

The following lemma shows that if~$O$ is non-empty, then~$x$ and~$y$ are opposite, that is, if any object is opposite to~$x$, then also~$y$ is opposite to~$x$ or the guess is wrong.
\begin{lemma}\label{lemma:candidateLists:xyopp}
Let~$\Phi := (x,y)$ be a guess. If~$|O| > 0$, then the guessed objects~$x$ and~$y$ are opposite or the guess~$\Phi$ is wrong.
\end{lemma}
\begin{proof}
Let~$q$ be an object in~$O$. That means that there exists a candidate list~$C(p,e) \in \mathcal{C}$ such that~$q \in C(p,e)$,~$q$ is decided in direction~$1-d(p,e)$ and~$q$ is opposite to~$p$. Let~$r$ be the other guessed object. We will now show that~$p$ and~$r$ must be opposite or the guess~$\Phi$ is wrong.
Note that since~$q$ is decided in direction~$1-d(p,e)$ it must be swapped with~$q$ at~$e$. However, before~$p$ meets~$q$, it is swapped with~$r$ in the first swap. Suppose that~$p$ and~$r$ are not opposite. Then~$p$ and~$r$ are only swapped once. However, since~$q$ is assigned the same direction as~$r$ and~$r$ reaches its destination before it meets~$p$ a second time, also~$q$ cannot meet~$p$ a second time and thus, there exists no valid selection that respects this guess and hence, the guess is wrong.
\end{proof}

We will now prove the following equality regarding the cardinality of~$\mathcal{C}$.

\begin{lemma}\label{lemma:candidateLists:n+O}
Let~$\Phi := (x,y)$ be a guess. If~$x$ and~$y$ are opposite, then~$|\mathcal{C}| = n + |O|$ or~$\Phi$ is wrong.
\end{lemma}
\begin{proof}
Let~$p \in \{x,y\}$ be a guessed object and let~$r \in \{x,y\}$ with~$p \neq r$ be the other guessed object. Suppose that
$p$ and~$r$ are opposite. We partition the set~$O$ into two subsets~$Q_p$ and~$Q_r$ as follows: Let~$q$ be an object such that there exists a candidate list~$C(p,e) \in \mathcal{C}$ such that~$q \in C(p,e)$,~$q$ is decided in direction~$1-d(p,e)$ and~$q$ is opposite to~$p$.
Then~$q \in Q_p$ and otherwise~$q \in Q_r$.
Observe that if~$q \in Q_r \subseteq O$, then~$q$ is in some candidate list~$C(r,e)$ and decided in direction~$1 - d(r,e)$.

Observe that~$p$ is swapped with~$r$ before it is swapped with any object in~$O$ because~$p$ and~$r$ are the guessed objects in~$\Phi$. Since~$p$ and~$r$ are opposite, they need to be swapped twice as otherwise the guess is wrong. Suppose the guess is not wrong. The union~$P_x \cup P_y$ covers the whole cycle~$C_n$ and hence, there exist at least~$n$ candidate lists, i.e.~$|\mathcal{C}| \geq n$. Note that after~$p$ is swapped with~$r$ it is swapped with every object in~$Q_p$ once, before reaching its destination. Since~$p$ and~$r$ are symmetric, we can derive the following. The guessed object~$p$ is swapped with~$|Q_p|$ objects after it is swapped with~$r$ the second time and guessed object~$r$ is swapped with~$|Q_r|$ objects after it is swapped with~$p$ the second time. From this it follows that the number of candidate lists is equal to~$n + |Q_p|+ |Q_r| = n + |O|$. This concludes the proof.
\end{proof}

The next lemma shows an equality involving the cardinalities of~$\mathcal{C}_D, D,$ and~$O$.
	
\begin{lemma}\label{lemma:c_d=d+o}
Let~$\Phi := (x,y)$ be a guess. Then either~$\Phi$ is wrong or it holds that
\begin{equation}\label{equation:lemma:c_d=d+o}
    |\mathcal{C}_D| = |D| + |O|.
\end{equation}
\end{lemma}
\begin{proof}
We start with the following observation which is true due to~$O \subseteq D$:
\begin{equation*}
    |D\setminus O| = |D| - |O|.
\end{equation*}
First observe that every candidate list~$C(p,e) \in \mathcal{C}_D$ contains exactly one object~$q \in D$ such that~$q$ is decided in direction~$1-d(p,e)$, due to \cref{lemma:candidateList:decidedByDecided}.

Now recall that every object~$q \in D \setminus O$ is non-opposite to both guessed objects~$x$ and~$y$ and can therefore appear in at most one candidate list of~$x$ and in one candidate list of~$y$, or otherwise~$\Phi$ is wrong. Further, since~$q \in D\setminus O$, object~$q$ is decided. However, since in every valid selection that respect~$\Phi$,~$\gamma(x) \neq \gamma(y)$,~$q$ is either decided in direction~$1-d(x,e_x)$ for some edge~$e_x$ such that~$q \in C(x,e_x)$ or~$q$ is decided in direction~$1-d(y,e_y)$ for some edge~$e_y$ such that~$q \in C(y,e_y)$ but not both.
Thus, the number~$\epsilon_0$ of candidate lists~$C(p,e)$ that contain an object~$q \in D \setminus O$ is equal to~$|D \setminus O|$.

Next, recall that every object~$q \in O$ is opposite to one guessed object~$p \in \{x,y\}$ for which there exists an edge~$e$ such that~$q \in C(p,e)$ and~$q$ is decided in direction~$1-d(p,e)$.
Hence~$q$ appears in two candidate list of~$p$ and in no candidate list of the other guessed object~$r \neq p$, or otherwise~$\Phi$ is wrong. 
Thus, the number~$\epsilon_1$ of candidate lists~$C(p,e)$ that contain an object~$q \in O$ is equal to~$2 |O|$.

We will now combine these results to show \cref{equation:lemma:c_d=d+o}.
Observe that $$|\mathcal{C}_D| = \epsilon_0 + \epsilon_1 = |D \setminus O| + 2|O| = |D| - |O| + 2|O| = |D| + |O|.$$
\end{proof}

The following lemma shows an inequality regarding the cardinality of~$\mathcal{C}_U$.

\begin{lemma}\label{lemma:candidateLists:o>0}
Let~$\Phi := (x,y)$ be a guess.
Then,~$|\mathcal{C}_U| \geq |U|$ or~$\Phi$ is wrong.
\end{lemma}
\begin{proof}
We make a case distinction over whether or not~$x$ and~$y$ are opposite.
If~$x$ and~$y$ are opposite, then due to \cref{lemma:candidateLists:n+O}, it holds that~$|\mathcal{C}| = n + |O|$, 	where~$n$ is the number of objects.
Recall that~$\mathcal{C}_D$ and~$\mathcal{C}_U$ are a partition of~$\mathcal{C}$ and thus~$|\mathcal{C}_D| + |\mathcal{C}_U| = n + |O|$.
Applying \cref{equation:lemma:c_d=d+o} yields~$|D| + |O| + |\mathcal{C}_U| = n + |O|$, which is equivalent to~$|\mathcal{C}_U| = n - |D|$.
Note that~$n = |U| + |D| + |D_0|$, since~$U, D$ and~$D_0$ are by \cref{obs:partition} a partition of~$X$ and hence~$|\mathcal{C}_U| \geq |U|$.

If~$x$ and~$y$ are non-opposite, then \cref{lemma:candidateLists:decidedByDestination,lemma:candidateLists:xyopp} imply that every object whose destination is not on the joint path~$P_x \cup P_y$ of~$x$ and~$y$ is in~$D_0$.
Hence, if an object is in~$U$ or~$D$, then its destination is on~$P_x \cup P_y$.
Note further that for an object there exists as many candidate lists as there are edges on that objects path.
Thus, the number objects in~$U$ and~$D$ are upper-bounded by the number of candidate lists of~$x$ and~$y$. From this follows that~$|\mathcal{C}| \geq |U| + |D|$.
By \cref{lemma:candidateLists:C_UandC_D}, $\mathcal{C}_D$ and~$\mathcal{C}_U$ are a partition of~$\mathcal{C}$.
Hence~$|\mathcal{C}_D| + |\mathcal{C}_U| \geq |U| + |D|$.
Applying \cref{equation:lemma:c_d=d+o} yields~$|D| + |O| + |\mathcal{C}_U| \geq |U| + |D|$, which is equivalent to~$|O| + |\mathcal{C}_U| \geq |U|$.
By \cref{lemma:candidateLists:xyopp} and since~$x$ and~$y$ are non-opposite,~$|O| = 0$ or~$\Phi$ is wrong.
Thus,~$|\mathcal{C}_U| \geq |U|$ or~$\Phi$ is wrong.
This concludes our proof.
\end{proof}

We will now conclude that the number of undecided objects in a candidate list of the guessed objects~$x$ and~$y$ is either zero or two.
	
\begin{proposition}
\label{proposition:candidateLists}
Let~$\Phi := (x,y)$ be a guess.
Let further~$(p,e) \in H$, such that~$p \in \Phi$.
The number of undecided objects in~$C(p,e)$ is either zero or two or~$\Phi$ is wrong.
\end{proposition}

\begin{proof}
We will prove the above statement by using a counting argument.
\Cref{lemma:candidateLists:o>0} shows that~$|\mathcal{C}_u| \geq |U|$ or~$\Phi$ is wrong.
Further in \cref{lemma:candidateLists:undecidedMax2} we showed that every undecided object is in exactly two candidate lists~$C(p,e)$ and~$C(q,e)$ where~$p,q \in \{x,y\}$.
Lastly, by definition of~$\mathcal{C}_U$ every candidate list~$C(p,e) \in \mathcal{C}_U$ contains at least two distinct objects~$q_0,q_1 \in U$. 

From this it follows that every candidate list~$C(p,e) \in \mathcal{C}_U$ contains exactly two undecided objects, as otherwise there exist not enough undecided objects to fill every candidate list in~$\mathcal{C}_U$ with at least two undecided objects, since every undecided object can be in at most two candidate lists in~$\mathcal{C}_U$.

The statement now follows from the fact that each candidate list in~$\mathcal{C}_U$ contains exactly two undecided objects and that each candidate list in~$\mathcal{C}_D$ contains no undecided objects.
\end{proof}

The last ingredient for the construction of a 2-SAT formula for exact selection is a generalized notion of decided objects, which will also cover objects that are decided, relative to the direction of other objects.
For this notion we will define an edge similar to the unique edge in \cref{lem:uniqueSwap}.
The difference being that where \cref{lem:uniqueSwap} states that there is a unique edge where two objects \emph{can be swapped}, we now define a unique edge where the two objects \emph{can meet}.
To this end, we first show that the number of objects moving into clockwise direction is the same for every valid selection.
		
\begin{lemma}\label{lemma:type}
Let~$\mathcal{I} := (N, X, \succ, C_n, \sigma_0, \sigma)$ be a instance of \ra{} on cycle~$C_n$.
For the set of objects~$X$ one can calculate in~$O(n)$ time the number of objects moving in clockwise and respectively counter-clockwise direction in every valid selection.
\end{lemma}
\begin{proof}
Let the variables~$d_1,...,d_n \in \{0, 1\}$ represent objects~$x_1,...,x_n$ where~$x_i$ moves in clockwise direction if~$d_i = 1$ and counter-clockwise otherwise.
Let further~$y_1,...,y_n \in \mathbb{N}$ denote the length of the path of an object's initial position to its destination in clockwise direction.
Since for each edge on an object's path in a certain direction there must be an object moving in the opposite direction swapping the object at that edge the sum of the lengths of the paths of the objects moving in clockwise direction must be equal to the sum of the lengths of the paths of the objects moving in counter-clockwise direction.
Formally,~$\sum_{i=1}^n d_i y_i = \sum_{i=1}^n ((1-d_i) (n - y_i))$, which can be rewritten as~$\sum_{i=1}^n (d_i y_i - (1-d_i) (n - y_i)) = 0$.
This is equivalent to
\begin{equation*}
	-\sum_{i=1}^n n + n \sum_{i=1}^n d_i + \sum_{i=1}^n y_i = 0.
\end{equation*}
Given the graph, the objects, and the assignments~$\sigma_0$ and~$\sigma$, ~$\sum_{i=1}^n y_i$ is a constant~$Y$.
This constant can be computed in linear time by measuring for each object~$p$ the distance between~$\sigma_0^{-1}(p)$ and~$\sigma_0^{-1}(p)$ in clockwise direction. 
Using~$Y$ and~$\sum_{i=1}^n n = n^2$ yields
\begin{equation}
\sum_{i=1}^n d_i = \dfrac{n^2 - Y}{n} = \dfrac{n^2}{n} - \dfrac{Y}{n} = n - \dfrac{Y}{n}
\end{equation}

Hence, the number of objects moving in clockwise direction is equal to~$n - \dfrac{Y}{n}$ and the number of objects moving in counter-clockwise direction is~$\dfrac{Y}{n}$.
\end{proof}

We will henceforth denote the number of objects walking in clockwise direction in instance~$\mathcal{I}$ by~$\theta_\mathcal{I}$ or~$\theta$ for short if the instance is implicitly clear.
We will now use~$\theta_\mathcal{I}$ to show that there is a unique edge where two agents can meet in all valid selections.

\begin{lemma}
\label{lemma:simulation}
Let~$\mathcal{I} := ((N, X, \succ, C_n, \sigma_0, \sigma), e')$ be an instance of \fsra{}.
Let~$\Phi := (x,y)$ be the guess of~$\mathcal{I}$, let~$(p,e) \in H$ be an object-edge-pair and let~$q \in C(p,e)$ where~$q$ is an undecided object.
Let~$P$ be the shared path of~$p$ and~$q$ such that~$e \in E(P)$.
After an~$O(n^3)$-time preprocessing, we can compute in constant time an edge~$f$ on~$P$ such that for every valid selection~$\gamma$ that respects~$\Phi$ and for which~$\gamma(q) \neq \gamma(p) = d(p,e)$ it holds that~$p$ and~$q$ meet at edge~$f$ in the execution of \texttt{Greedy Swap} on input~$\gamma$.
\end{lemma}

\begin{proof}
Let~$\Phi := (x,y)$ be the guess of~$\mathcal{I}$, let~$(p,e) \in H$ be an object-edge-pair and let~$q \in C(p,e)$ be an undecided object.
Let~$P$ be the shared path of~$p$ and~$q$ such that~$e \in E(P)$.
We will assume without loss of generality that~$\gamma(p) = \gamma(x) = 1$. 
The other case is symmetric with~$x$ and~$y$ and clockwise and counter-clockwise swapped.

We will first show that after computing the set of all decided objects, we can compute in constant time the constant number~$\lambda(p,x)$ of objects that start on~$\cycseq[p][x]$ that are assigned clockwise direction and the constant number~$\lambda(x,p)$ of objects that start on~$\cycseq[x][p]$ that are assigned clockwise direction in every valid selection~$\gamma$ with~$\gamma(p) = \gamma(x) \neq \gamma(y) = 0$.
Observe that~$\lambda(p,y) + \lambda(y,p) = \theta + 2$ as both~$p$ and~$x$ walk in clockwise direction and are counted in both~$\lambda(p,y)$ and~$\lambda(y,p)$.
Thus, given~$\theta$ and either~$\lambda(p,x)$ or~$\lambda(x,p)$, we can compute the respective other term in constant time.
Afterwards we show how to compute the edge~$f$ over which~$p$ and~$q$ can be swapped in any valid sequence of swaps that respects~$\Phi$.

Let $\xi_\gamma(p,y)$ be the set of shared paths of~$p$ and~$y$.
We will distinguish between three cases: (i)~$\xi_\gamma(p,y)$ is empty, (ii)~$\xi_\gamma(p,y)$ is non-empty but~$p$ does not appear in a candidate list of~$y$, or (iii)~$\xi_\gamma(p,y)$ is non-empty and~$p$ appears in a candidate list of~$y$.
By \cref{lemma:decidedObjects}, $p$ is in the first case decided.
If it is decided in direction~$0$, then~$\gamma$ is not valid as~$\gamma(p) = 1$.
Since~$\lambda(p,y)$ and~$\lambda(y,p)$ are only defined for valid selections we can ignore this case.
If~$p$ is decided in direction~$1$, then due to \cref{lemma:between_decided}, every object that starts on the path between~$p$ and~$x$ in clockwise direction is decided as well.
Hence, the number $\lambda(x,p)$ is equal to the number of those objects decided in clockwise direction (including~$x$ and~$p$) and can be computed in linear time.

	In the second case, $\xi_\gamma(p,y)$ is non-empty but~$p$ does not appear in a candidate list of~$y$.
Then $\gamma$ is not valid and we can again ignore this case.
In the last case~$\xi_\gamma(p,y)$ is non-empty and~$p$ appears in a candidate list of~$y$.
We will now prove that for any valid selection~$\gamma$ where~$\gamma(p) = \gamma(x)$,~$\lambda(p,x)$ is equal to the number of edges that~$y$ has to pass before reaching the edge~$e_y$ where~$p$ and~$y$ can be swapped according to \cref{lem:uniqueSwap}.
Since~$p$ is in the candidate list of~$y$ we know that~$e_y$ must exist and that~$p$ and~$y$ have to be swapped at~$e_y$ or otherwise~$\gamma$ is not valid.
We know further that~$y$ has to be swapped with~$\lambda(p,x)$ objects (including~$x$) before it meets~$p$, as otherwise it cannot be swapped with~$p$ over~$e_y$.
Suppose towards a contradiction that~$\lambda(p,x)$ is not equal to this number and~$\gamma$ is valid.
Because~$y$ has to be swapped with~$\lambda(p,x)$ objects---as this many objects move in clockwise direction before~$y$ meets~$p$---and each swap amounts to one edge being passed,~$y$ and~$p$ meet at an edge~$e'$ that is not~$e_y$.
However, \cref{lem:uniqueSwap} states that~$e_y$ is the only edge on the shared path of~$p$ and~$y$ where they can be swapped.
Thus,~$p$ and~$y$ cannot be swapped at edge~$e'$ and~$\sigma$ cannot be reached and therefore~$\gamma$ cannot be valid, a contradiction.
Since edge~$e_y$ is constant for all valid selections, so is~$\lambda(p,x)$.

We conclude this proof by showing how to compute the edge~$f$ where~$p$ and~$q$ can meet in every valid selection~$\gamma$ that respects~$\Phi$ and with~$\gamma(q) \neq \gamma(p) = \gamma(x)$ in constant time after an~$O(n^3)$-time preprocessing step.
The set of decided objects for the guess~$\Phi$ can be computed in~$O(n^2)$ time.
Doing this for each guess then takes~$O(n^3)$ time.
Afterwards we can compute a table~$T[i,j]$ that stores how many objects that start between~$i$ and~$j$ are decided in clockwise direction.
Computing the first entry~$T(i,i+1\bmod n)$ takes~$O(1)$ time and afterwards we can compute another entry~$T[i,j]$ in constant time by checking whether~$y+1 \bmod n$ is decided in clockwise direction.
Thus computing all entries for~$i$ takes~$O(n)$ time and computing all entries of the table takes~$O(n^2)$ time.
Afterwards we can compute the number~$\lambda(p,x)$ in constant time and by comparing~$\lambda(p,x)$ and~$\lambda(q,y)$, we can calculate in constant time the number of objects between~$p$ and~$q$ in clockwise direction that move in clockwise direction.
Counting this number of edges from~$q$ in counter-clockwise direction yields the edge~$f$ where~$p$ and~$q$ can meet if this many objects between them move in clockwise direction and hence are swapped with~$q$.
\end{proof}

We use \cref{lemma:simulation} to define successful tuples as follows.
A tuple is successful if the edge where two objects meet according to \cref{lemma:simulation} is the edge over which they can be swapped according to \cref{lem:uniqueSwap}.
\begin{definition}\label{def:simulations}
Let~$(p,e) \in H$ be an object-edge-pair and let~$q \in C(p,e)$ such that ~$q$ is undecided.
Let~$P$ be their shared path and let~$e \in E(P)$.
Let~$f$ be the edge according to \cref{lemma:simulation} where~$p$ and~$q$ meet in every valid selection~$\gamma$ that respects~$\Phi$ and for which it holds that~$\gamma(q) \neq \gamma(p) = d(p,e)$.
If~$q$ is closer to~$p$ in direction~$d(p,e)$ than guessed object~$\Phi_{d(p,e)}$, then let~$c := 1$ and~$c := 0$ otherwise. 
Then,~$S(p,e,q)$ denotes the tuple~$(f,c)$ and one says that it is \emph{successful} if~$e=f$ and unsuccessful otherwise.
\end{definition}

We will now introduce conditionally decided objects and construct the 2-SAT formula afterwards.

\begin{definition}\label{def:p-decidedObject}
Let~$\mathcal{I} := ((N, X, \succ, C_n, \sigma_0, \sigma), e)$ be an instance of \fsra{}.
Let~$C(p,e)$ be a candidate list and let $q \in C(p,e)$ be an object such that there exists a direction~$d$ such that there is no selection~$\gamma$ with~$\gamma(p) = d(p,e)$ and $\gamma(q) = d$ that respects~$\mathcal{I}$ and that is valid.
Then~$q$ is \emph{$p$-decided} in direction~$1-d$.
\end{definition}
Observe that the class of decided objects is a sub-class of $p$-decided objects.
Analogously to decided objects, we present rules to find conditionally decided objects and say that objects are conditionally undecided if none of the rules apply for them.
Let~$C(p,e)$ be a candidate list, let~$q \in C(p,e)$ be an object, and let~$d$ be the direction such that~$e$ is on the path of~$p$ in direction~$d$.
Then~$q$ is~$p$-decided in direction~$1-d$ if one of the following rules apply:
\begin{enumerate}
	\item The object~$q$ shields $p$ in direction $d(p,e)$.
	\item The tuple~$S(p,e,q)$ is unsuccessful.
	\item There exists a candidate list~$C(p,e)$ such that $q \in C(p,e)$ and there exists an object $q \neq q' \in C(p,e)$, such that $q'$ is decided in direction $1-d(p,e)$.
\end{enumerate}

We now show the correctness of the three rules for~$p$-decided objects.
The first rule focuses on objects that shields $p$ in direction $d(p,e)$.

\begin{lemma}\label{lem:decidedShield}
Let $(p,e) \in H$ be an object-edge-pair and let $q \in C(p,e)$ be an object in the corresponding candidate list. If $q$ shields $p$ in direction $d(p,e)$, then $q$ is $p$-decided in direction $1-d(p,e)$.
\end{lemma}
\begin{proof}
By definition, there exists no harmonic selection $\gamma$ such that $\gamma(p) = d(p,e)$ and $\gamma(q) = d(p,e)$ if $q$ shields $p$ in direction $d(p,e)$. Thus, $q$ is $p$-decided in direction $1-d(p,e)$.
\end{proof}

The second rule focuses on objects~$q \in C(p,e)$ where~$S(p,e,q)$ is unsuccessful.

\begin{lemma}\label{lem:decidedUnsuccessful}
Let $\Phi$ be a guess, let $(p,e) \in H$ be an object-edge-pair and let $q \in C(p,e)$ be an object such that $S(p,e,q)$ is unsuccessful. Then $q$ is $p$-decided in direction $d(p,e)$.
\end{lemma}
\begin{proof}
Let $P$ be the shared path of $p$ and $q$ that contains $e$.
Since $S(p,e,q)$ is unsuccessful, for every valid selection $\gamma$ that respects $\Phi$ and where $\gamma(p) = d(p,e) = 1- \gamma(q)$, $p$ and $q$ meet at an edge $f \neq e$ that is on $P$, according to \cref{lemma:simulation}. However, \cref{lem:uniqueSwap} tells us that since $p$ and $q$ can be swapped over $e$ (as $q$ is in the candidate list $C(p,e)$), $p$ and $q$ cannot be swapped over $f$. Since $f$, however, is on the shared path of $p$ and $q$, at least one of $p$ and $q$ cannot reach its destination. Hence, $\gamma$ is not valid, a contradiction. From this follows that there exists no such valid $\gamma$ and thus, $q$ is $p$-decided in direction $d(p,e)$.
\end{proof}

Lastly, the third rule focuses on the effect that $p$-decided objects have on other objects in the same canidate list.

\begin{lemma}\label{lemma:pdecided_larger_one}
Let $\Phi$ be a guess and let $(p,e) \in H$ be an object-edge-pair. If $C(p,e)$ contains at least one object $q$ that is $p$-decided in direction $1-d(p,e)$, then every object $q' \in C(p,e) \setminus \{q\}$ is $p$-decided in direction $d(p,e)$ or $\Phi$ is wrong.
\end{lemma}
\begin{proof}
Suppose towards a contradiction that $C(p,e)$ contains at least one object $q$ that is $p$-decided in direction $1-d(p,e)$ but there exists at least one object $q' \neq q$ such that $q'$ is not $p$-decided in direction $d(p,e)$.

However, since $q$ is $p$-decided in direction $1-d(p,e)$ there exists no valid selection $\gamma$ that respects $\Phi$ such that $\gamma(p) = d(p,e) = \gamma(q)$. Suppose towards a contradiction that there exists a valid selection $\gamma$ that respects $\Phi$ such that $\gamma(q') =\gamma(q) = 1-d(p,e) =1- \gamma(p)$. But then, $f_\gamma(p,e) > 1$ and thus, $\gamma$ is not exact, a contradiction.
From this follows that there exists no valid selection $\gamma$ that respects $\Phi$ where $\gamma(p) = d(p,e) = 1-\gamma(q')$ and thus, $q'$ is $p$ decided in direction $d(p,e)$. If $q'$, however, is already known to be $p$-decided in direction $1-d(p,e)$ then there exists no valid selection that respects $\Phi$ and thus, $\Phi$ is wrong.
\end{proof}

We finally construct a 2-SAT formula for exact selections using decided, conditionally decided and undecided objects.

\begin{construction}
\label{const:exact}
Let~$\Phi := (x,y)$ be the guess of an instance~$\mathcal{I} := ((N, X, \succ, C_n, \sigma_0, \sigma), e')$ of \fsra{} and let~$C(p,e) \in H$ be a candidate list.

If the number of $p$-decided objects in direction $1-d(p,e)$ is not one and there exists no $p$-undecided object in $C(p,e)$, then we add the clause
\begin{equation}\label{eq:exact:pneqd}
(p \neq d(p,e)).\footnotemark
\end{equation}
\footnotetext{We use the notation~$q \neq d(p,e)$ to avoid case distinctions. Since~$d(p,e)$ is precomputed, this clause is equivalent to~$\neg q$ if~$d(p,e)=1$ and~$q$ otherwise.}

Otherwise, for every decided object $q$ in direction $d$ we add the clause 
\begin{equation}\label{eq:exact:qeqd}
q = d.
\end{equation}
Further, for every $p$-decided object $q$ in direction $d$ that is not also a decided object, we add the clause
\begin{equation}\label{eq:exact:peqd_implies_qeqd}
(p = d(p,e)) \rightarrow (q = d).
\end{equation}

Lastly, if $p \in \{x,y\}$ and there exists an undecided object $q \in C(p,e)$, then we distinguish between two cases. If there exists exactly one other object $q' \in C(p,e)$ that is undecided, then we add the clause
\begin{equation}\label{eq:exact:qneq_prime}
q \neq q'.
\end{equation}
Otherwise we add the the clause
\begin{equation}\label{eq:exact:bot}
\bot.
\end{equation}

We use~$\psi_e$ to denote the constructed formula and~$\phi = \psi_h \land \psi_e$. 
\end{construction}

Before we can show that \cref{const:exact} is correct, we show two final lemmas.
The first lemma states that if two objects appear in some candidate list and their respective tuples are both successful, then they also appear in the same candidate list of a guessed object.
\begin{lemma}\label{lemma:simulation:sameCandidateList}
Let~$p, q_0$ and~$q_1$ be three objects and let~$e$ be an edge. If~$S(p,e,q_0) = S(p,e,q_1)$, then~$q_0$ and~$q_1$ are in the same candidate list of guessed object~$\Phi_{d(p,e)}$.
\end{lemma}
\begin{proof}
Let~$(f,c) = S(p,e,q_0) = S(p,e,q_1)$.
Let~$e_0, e_1$ be the edges such that for every valid selection~$\gamma$ that respects~$\Phi$ and with~$\gamma(p) = d(p,e) = 1-\gamma(q_0) = 1 - \gamma(q_1)$ it holds that~$p$ and~$q_0$ meet at edge~$e_0$ and~$p$ and~$q_1$ meet at edge~$e_1$ according to \cref{lemma:simulation}.
By definition of~$f$, it holds for every valid selection~$\gamma$ that respects~$\Phi$ and with~$\gamma(p) = d(p,e)$ that if~$\gamma(q_0) \neq \gamma(p)$, then~$p$ and~$q_0$ meet at edge~$f$ and if~$\gamma(q_q) \neq \gamma(p)$, then~$p$ and~$q_1$ meet at edge~$f$ according to \cref{lemma:simulation}.
By definition of~$c$, either~$q_0$ and~$q_1$ are both closer to~$p$ in direction~$d(p,e)$ than the guessed object~$\Phi_{d(p,e)}$ or both further from~$p$ in direction~$d(p,e)$ than~$\Phi_{d(p,e)}$.

We distinguish between two cases.
In the first case~$c = 1$.
Then, according to \cref{lemma:simulation}, there exists a number~$\lambda$ of objects such that~$q_0$ and~$q_1$ are both swapped with exactly~$\lambda$ objects after being swapped with~$p$ and before meeting guessed object~$\Phi_{d(p,e)}$, at an edge~$e^*_0$ and~$e^*_1$ respectively.

In the second case~$c = 0$. Then, according to \cref{lemma:simulation}, there exists a number~$\lambda$ of objects such that~$q_0$ and~$q_1$ are both swapped with exactly~$\lambda$ objects after being swapped with guessed object~$\Phi_{d(p,e)}$, at an edge~$e^*_0$ and~$e^*_1$ respectively, and before meeting~$p$.

Since~$\lambda$ is the same constant for both~$q_0$ and~$q_1$ in both cases, it holds that~$e^*_0 = e^*_1$.
In \cref{lemma:simulation},~$e^*_0$ and~$e^*_1$ are chosen to be the edges where~$\Phi_{d(p,e)}$ can be swapped with~$q_0$ and~$q_1$ respectively. Thus,~$q_0$ and~$q_1$ are in the same candidate list of~$\Phi_{d(p,e)}$.
\end{proof}

The next lemma shows the relevance of the second entry~$c$ in~$(f,c) = S(p,e,q)$.
Using this lemma, we can always eliminate all candidates with~$c=1$ or all candidates with~$c=0$.

\begin{lemma}\label{lemma:remaining:partition}
Let $\Phi$ be a guess, let~$(p,e) \in H$ be an object-edge-pair and let~$Q$ be the set of undecided objects~$q \in C(p,e)$ for which~$(f_q,c_q) = S(p,e,q)$ is successful.
Then~$Q$ can be partitioned into two sets~$Q_0$ and~$Q_1$ where~$q$ belongs to~$Q_{c_q}$ and the following statement holds.
There exists an~$i \in \{0,1\}$ such that if~$p$ is not decided in direction~$1-d(p,e)$, then for every object~$q \in Q_i$~$q$ is decided in direction~$d(p,e)$.
\end{lemma}
\begin{proof}
Let~$\xi_\gamma(p,\Phi_{1-d(p,e)})$ be the set of shared path for every selection $\gamma$ where $\gamma(p) = d(p,e) = 1 - \gamma(\Phi_{1-d(p,e)})$.
We distinguish between two cases.
In the first case $\xi_\gamma(p,\Phi_{1-d(p,e)})$ is empty and hence~$p$ does not occur in a candidate list of~$\Phi_{1-d(p,e)}$.
According to \cref{lemma:decidedObjects},~$p$ is decided.
If~$p$ is not decided in direction~$1-d(p,e)$, then according to \cref{lemma:between_decided}, every object that starts between $\Phi_{d(p,e)}$ and~$p$ in direction $d(p,e)$ is decided.
Since $p$ cannot be swapped with any of these objects before being swapped with~$\Phi_{1-d(p,e)}$ (with which it is never swapped), none of these objects can be swapped with~$p$ and if they are contained in~$C(p,e)$, then they are decided in direction~$d(p,e)$.
Thus, every object $q \in Q_0$ is decided in direction~$d(p,e)$.

In the second case $\xi_\gamma(p,\Phi_{1-d(p,e)})$ is non-empty but~$p$ does not appear in a candidate list of $\Phi_{1-d(p,e)}$.
Then again~$p$ is decided and we can repeat the argument above.
In the last case~$\xi_\gamma(p,\Phi_{1-d(p,e)})$ is non-empty and~$p$ appears in a candidate list of $\Phi_{1-d(p,e)}$.
Let~$e^*$ be the edge where~$p$ and~$\Phi_{1-d(p,e)}$ can be swapped and which is on the same shared path of~$p$ and~$\Phi_{1-d(p,e)}$ as~$e$.
We distinguish between two cases.
In the first case,~$e^*$ is closer to~$p$ in direction~$d(p,e)$ than~$e$ and since every object~$q_1 \in Q_1$ is closer to~$p$ in direction~$d(p,e)$ than~$\Phi_{1-d(p,e)}$, according to \cref{order_lemma},~$\Phi_{1-d(p,e)}$ is shielded by~$q_1$ in direction~$1-d(p,e)$.
Thus, every~$q_1 \in Q_1$ is decided in direction~$d(p,e)$.
In the second case,~$e^*$ is further away from~$p$ in direction~$d(p,e)$ than~$e$.
Since every object~$q_0 \in Q_0$ is further away from~$p$ in direction~$d(p,e)$ than~$\Phi_{1-d(p,e)}$, according to \cref{order_lemma},~$\Phi_{1-d(p,e)}$ shields~$q_0$ in direction~$1-d(p,e)$.
Thus, every~$q_0 \in Q_0$ is decided in direction~$d(p,e)$.
\end{proof}

Based on \cref{lemma:simulation:sameCandidateList,lemma:remaining:partition}, we can finally prove that a selection is valid if and only if it corresponds to a satisfying truth assignment to~$\phi$. 

\begin{proposition}
\label{proposition:final}
Let~$\Phi := (x,y)$ be the guess of an instance~$\mathcal{I} := ((N, X, \succ, C_n, \sigma_0, \sigma), e')$ of \fsra{}.
Let~$\gamma$ be a selection that respects~$\Phi$.
Then~$\gamma$ is valid if and only if it corresponds to a satisfying truth assignment of~$\phi$.
\end{proposition}

\begin{proof}
We show that a selection that respects~$\Phi$ is valid if and only if it corresponds to a satisfying truth assignment of~$\phi$.
We first show that if a selection is valid, then it corresponds to a satisfying truth assignment of~$\phi$ and afterwards we show that if a selection corresponds to a satisfying truth assignment of~$\phi$, then it is valid.

For the first step, assume towards a contradiction that $\gamma$ is valid and respects $\Phi$ but it does not correspond to a satisfying truth assignment of $\phi$.
Then there is a clause~$c$ in~$\phi$ that is not satisfied by the truth assignment corresponding to~$\gamma$.
Notice that~$c$ cannot be a clause in~$\psi_h$ as~$\gamma$ is valid and \cref{harmonicSel} states that a selection is harmonic if and only if it corresponds to a satisfying truth assignment of~$\psi_h$.
Thus~$c$ is a clause in~$\psi_e$.
We distinguish between the six different kinds of clauses in~$\psi_e$.

In the first case $c$ is a clause as in \cref{eq:exact:pneqd}.
Then there exists an object-edge-pair $(p,e) \in H$ such that $C(p,e)$ does not contain
exactly one $p$-decided object in direction $1-d(p,e)$, nor a $p$-undecided object. However, since $c$ is not fulfilled by $\gamma$, it holds that $\gamma(p) = d(p,e)$.
We distinguish between two cases. In the first case there exists no object $q \in C(p,e)$ for which there exists a valid selection $\gamma'$ where $\gamma'(p) = d(p,e) = 1- \gamma'(q)$ and thus, $\gamma$ cannot be valid, as then $f_\gamma(p,e) = 0$, a contradiction.
In the second case there exist at least two objects $q_0,q_1 \in C(p,e)$ such that for every valid selection $\gamma'$ where $\gamma'(p) = d(p,e)$ it holds that $\gamma(q_0) = \gamma(q_1) = 1-d(p,e)$ and thus, $\gamma$ cannot be valid, as then $f_\gamma(p,e) > 1$, a contradiction.

In the second case $c$ is a clause as in \cref{eq:exact:qeqd}.
Then there exists an object $q \in C(p,e)$ that is decided in a direction $d$ but since $c$ is not fulfilled in $\gamma$, $\gamma(q) = 1-d$. Since according to \cref{def:decidedObject}, there exists no selection $\gamma'$ that respects $\Phi$ and where $\gamma'(q) = 1-d$, $\gamma$ cannot be valid, a contradiction. 

In the third case $c$ is a clause as in \cref{eq:exact:peqd_implies_qeqd}. Then there exists an object $q \in C(p,e)$ that is $p$-decided in a direction $d$ but since $c$ is not fulfilled in $\gamma$, $\gamma(q) = 1-d$ and $\gamma(p) = d(p,e)$.
Since according to \cref{def:p-decidedObject}, there exists no selection $\gamma'$ that respects $\Phi$ and where $\gamma'(q) = 1-d$ and $\gamma(p) = d(p,e)$, $\gamma$ cannot be valid, a contradiction. 

In the fourth case $c$ is a clause as in \cref{eq:exact:qneq_prime}. Then there exist objects $q, q' \in C(p,e)$ that are undecided
and there exists no other undecided object $q'' \in C(p,e)$ such that $q \neq q'' \neq q'$ and since $c$ is not fulfilled in $\gamma$, $\gamma(q) = \gamma(q')$ and thus, either $f_\gamma(p,e) = 0$ or $f_\gamma(p,e) = 2$, a contradiction.

In the fifth case $c$ is a clause as in \cref{eq:exact:bot}. Then, however, there exists an object-edge-pair $(p,e) \in H$ such that $p \in \Phi$, an undecided object $q \in C(p,e)$ and not exactly one other undecided object in $C(p,e)$. But then, according to \cref{proposition:candidateLists}, $\Phi$ is wrong. From this follows that there exists no valid selection that respects $\Phi$ and thus, $\gamma$ cannot be valid and respect $\Phi$, a contradiction.

We now show that if a selection corresponds to a satisfying truth assignment of~$\phi$, then it is valid.
Suppose towards a contradiction that $\gamma$ corresponds to a satisfying truth assignment of $\phi$ but $\gamma$ is not valid.
By \cref{harmonicSel}, $\gamma$ has to be harmonic as it corresponds to a satisfying truth assignment of~$\psi_h$.
Thus,~$\gamma$ is not exact, that is, there exists a $(p,e) \in H$ such that $f_\gamma(p,e) \neq 1$.
We distinguish between four cases: (i)~$C(p,e)$ contains only~$p$-decided objects in direction~$d(p,e)$, (ii)~$C(p,e)$ contains only~$p$-decided objects and at least one~$p$-decided object in direction~$1-d(p,e)$, (iii)~$C(p,e)$ contains two~$p$-undecided objects~$q_0,q_1$ where one of them is closer to~$p$ than~$\Phi_{d(p,e)}$ in direction~$d(p,e)$ and the other is further away, and (iv)~$C(p,e)$ contains only~$p$-undecided objects that are all closer to~$p$ than~$\Phi_{d(p,e)}$ or are all further away.

In the first case $\phi$ contains the clause from \cref{eq:exact:pneqd} and since $\gamma$ corresponds to a satisfying truth assignment of $\phi$, $\gamma(p) = 1-d(p,e)$. From this follows that $f_\gamma(p,e) = 1$, a contradiction.

In the second case, according to \cref{lemma:pdecided_larger_one} there exists exactly one object $q \in C(p,e)$ that is $p$-decided in direction $1-d(p,e)$ and each object $q' \in C(p,e) \setminus \{q\}$ is $p$-decided in direction $d(p,e)$. We will now assume that $\gamma(p) = d(p,e)$ as otherwise $f_\gamma(p,e) = 1$, a contradiction.
For every $q \in C(p,e)$, which is $p$-decided in direction $d$, $\phi$ contains the clause $q = d$ as in either \cref{eq:exact:qeqd} or \cref{eq:exact:peqd_implies_qeqd}.
Since there is exactly one object $q$ that is $p$-decided in direction $1-d(p,e)$ and since we assumed that $\gamma(p) = d(p,e)$, according to its definition, $f_\gamma(p,e) = 1$, a contradiction.

In the third and fourth case~$C(p,e)$ contains at least one $p$-undecided object.
Note that by \cref{lemma:pdecided_larger_one}, if~$C(p,e)$ contains~$p$-undecided objects, then it does not contain~$p$-decided objects in direction $1-d(p,e)$ and if it contains only one object, then this object is~$p$-decided.
Observe further that by \cref{lem:decidedUnsuccessful} for each~$p$-undecided object~$q\in C(p,e)$ it holds that~$S(p,e,q)$ is successful.

In the third case, let~$q_0 \neq q_1 \in C(p,e)$ be two~$p$-undecided objects such that~$q_0$ be closer to~$p$ in direction~$d(p,e)$ than~$\Phi_{d(p,e)}$ and~$q_1$ be further away from~$p$ in direction~$d(p,e)$ than~$\Phi_{d(p,e)}$.
By \cref{lemma:remaining:partition}, if~$p$ is not decided in direction~$1-d(p,e)$, then either~$q_0$ or~$q_1$ is~$p$-decided in direction~$d(p,e)$, a contradiction to~$q_0$ and~$q_1$ being~$p$-undecided.
If~$p$ is decided in direction~$1-d(p,e)$, then~$f_\gamma(p,e) = 1$, a contradiction.
	
In the fourth case, we assume without loss of generality that~$C(p,e)$ contains at least two~$p$-undecided objects~$q_0,q_1$ as the first case already deals with empty candidate lists and there can not be a single~$p$-undecided object in a candidate list of~$p$ as shown above.
Further,~$q_0$ and~$q_1$ are either both closer to~$p$ in direction~$d(p,e)$ than~$\Phi_{d(p,e)}$ or both further away from~$p$ in direction~$d(p,e)$ than~$\Phi_{d(p,e)}$.
By definition,~$S(p,e,q_0) = S(p,e,q_1)$ in this case and according to \cref{lemma:simulation:sameCandidateList},~$q_0$ and~$q_1$ are in the same candidate list of~$\Phi_{d(p,e)}$.
If~$|C(p,e)| > 2$, then there is a candidate list~$C(\Phi_{d(p,e)},e')$ of~$\Phi_{d(p,e)}$ with~$|C(\Phi_{d(p,e)},e')| \geq |C(p,e)| > 2$ and hence~$\phi$ contains the clause~$\bot$ from~\cref{eq:exact:bot}.
This contradicts the fact that~$\gamma$ corresponds to a satisfying truth assignment of~$\phi$.
If~$|C(p,e)| = 2$, then~$\phi$ contains the clause~$q_0 \neq q_1$ from \cref{eq:exact:qneq_prime} and thus exactly one of~$q_0$ and~$q_1$ are assigned to direction~$1-d(p,e)$, regardless of the assigned direction of~$p$.
Thus~$f_\gamma(p,e) = 1$, a contradiction.
This concludes the proof.
\end{proof}

Based on \cref{proposition:final,observation:num.q.inCL}, we can show our main theorem.

\begin{theorem}
\label{thm:polyAlg}
\ra{} on cycles is decidable in~$\mathcal{O}(n^3)$ time.
\end{theorem}
\begin{proof}
Let~$\mathcal{I} := (N, X, \succ, C_n, \sigma_0, \sigma)$ be an instance of \ra.
Given the instance, we compute the 2-SAT formula~$\psi_h$ according to \cref{const:harmonic}.
We first precompute in overall~$\mathcal{O}(n^3)$ time for each pair of objects whether they are opposite and for each possible assignment of directions for these two objects whether one objects shields the other object and whether they are compatible.
Constructing~$\psi_h$ then takes~$\mathcal{O}(n^2)$ time as it requires to check for each pair~$p,q$ of objects whether they are compatible and whether one shields the other in any direction.
Afterwards we divide~$\mathcal{I}$ into~$n$ instances~$(\mathcal{I}, e)$ of \fsra.
This further determines the guess~$\Phi:=(x,y)$.
We then compute the set of all decided objects according to \cref{lemma:decidedOpposite,lemma:decidedObjects,lemma:candidateList:decidedByDecided,lemma:candidateLists:decidedByDestination,lemma:between_decided} in~$O(n^2)$ time.
Moreover, we iterate over all objects~$p$ and compute for each possible direction~$d$ of~$p$ the set of all~$p$-decided objects according to \cref{lem:decidedShield,lem:decidedUnsuccessful,lemma:pdecided_larger_one} in~$O(n)$ time.
Notice that \cref{proposition:candidateLists,lemma:simulation:sameCandidateList} imply that each candidate list~$C(p',e')$ contains at most two objects that are~$p'$-undecided.
Hence, performing \cref{lemma:pdecided_larger_one} can be performed for all objects that are not yet discovered to be~$p'$-decided in constant time per candidate list.
For a given guess, this then takes~$O(n^2)$ time for all objects.
We then create a 2-SAT formula according to~\cref{const:exact} to check for an exact selection for the given instance of \fsra{}.
For constructing~$\phi$, we add clauses for each object-edge-pair and due to \cref{observation:num.q.inCL}, there are~$\mathcal{O}(n^2)$ such clauses.
The time for computing each clause is constant and solving~$\phi$ takes time linear in the number of clauses, that is,~$O(n^2)$ time.
Thus, the procedure takes overall~$\mathcal{O}(n^2)$ time per instance of \fsra{} and~$\mathcal{O}(n^3)$ time in total.
\Cref{proposition:final} shows correctness.
\end{proof}

\section{NP-hardness of Reachable Assignment on Cliques}\label{chapter_cliques}
In this section we prove that \ra{} is NP-hard on cliques.
To do so, we will adapt the reduction of \ro{} to \ra{} by Gourves et al.\,\cite{gourves_object_2017} to preserve the property that the resulting graph is a clique if the input graph was a clique and use the fact that \ro{} is NP-hard on cliques \cite{bentert_good_2019}.

\begin{proposition}
\ra{} is NP-hard on cliques.
\end{proposition}
\begin{proof}
Gourves et al.\,\cite{gourves_object_2017} already described a reduction of \ro{} to \ra{} on general graphs.
Given an instance~$\mathcal{I} = (N, X, \succ, G=(N,E), \sigma_0, i, x_\ell)$ of \ro{}, they create an equivalent instance~$\mathcal{I}':=(N \cup N', X \cup X', {\succ',} G', \sigma_0', \sigma')$ of~\ra{} as follows.
They first copy all agents with their initial object, that is,~$N'=\{j' \mid j \in N\}, X' = \{x_j' \mid x_j \in X\}$, and~$\sigma_0' = \{(j,x_j),(j',x_j') \mid (j,x_j) \in \sigma_0\}$.
Then, they connect each agent to its copy and pairwise connect all of the copies, that is,~$G' = (N \cup N', E \cup \{(j,j')\mid j \in N\} \cup \{(j',k')\mid j' \neq k' \in N'\})$.
The preference lists of the original agents stay the same except that they now prefer the initial object of their respective copy the most.
The agent~$i'$ only prefers~$x_\ell$ over its initial object~$x_i'$.
Each other agent in~$X'$ only prefers objects in~$X \setminus \{x_\ell\}$ over its initial object and each agent~$j' \in X'$ prefers the object~$x_j$ the most except for~$\ell'$, which prefers~$x_i$ the most.
The remainder of the preference lists are then arranged such that if each agent in~$X' \setminus \{i'\}$ holds an arbitrary object from~$X \setminus \{x_\ell\}$, then there is always a sequence of rational swaps between the agents in~$X'$ such that each agent obtains its most preferred object.
Finally,~$\sigma'$ assigns each agent its most preferred object.

Observe that in order to reach~$\sigma'$, the agent~$i'$ has to obtain~$x_\ell$ from~$i$.
Since no agent in~$X' \setminus \{i'\}$ accepts~$x_\ell$, the object has to be swapped amongst the agents in~$N$, that is, the original instance has to be a yes-instance.
For the reverse direction, note that if~$x_\ell$ can be traded amongst the agents in~$N$ to agent~$i$, then after this sequence of swap each agent can swap its object with its respective copy and the copies can swap the objects such that each agent obtains its most preferred object.

We adapt this reduction to prove NP-hardness for \ra{} on cliques by requiring that the original graph is a clique and simply adding all of the missing edges between original agents and copies.
Bentert et al.\,\cite{bentert_good_2019} already proved NP-hardness for \ro{} on cliques.
The proof of correctness is then exactly the same as the original.
\end{proof}

\section{Conclusion}\label{chapter_conclusion}
In this work, we have investigated a version of house-marketing problems called \ra{}, a problem in the field of \emph{Multi-Agent-Systems} that was first proposed by Gourves et al.\,\cite{gourves_object_2017}.
We presented an~$\mathcal{O}(n^3)$-time algorithm for cycles and showed NP-hardness for cliques.

The key to solving \ra{} on trees and on cycles was to exploit the number of unique paths an object can be swapped along. 
Finding graph classes in which this number is bounded and solving \ra{} for these graph classes is a natural next step for further research.
Moreover, since cycles are paths with one additional edge, it seems promising to investigate graphs with constant feedback edge number.
Afterwards, one may study the parameterized complexity of \ra{} with respect to the parameter feedback edge number.
Other possibilities for parameters are related to the agent's preferences as studied by Bentert~et~al.\,~\cite{bentert_good_2019} for \ro{}.
One might also consider generalized settings such as allowing ties in the preference lists as studied by Huang and Xiao \cite{DBLP:conf/aaai/HuangX19} for \ro{}.

Finally, we mention that Katar{\'{\i}}na Cechl{\'{a}}rov{\'{a}} and Ildik{\'{o}} Schlotter \cite{DBLP:conf/iwpec/CechlarovaS10} studied the parameterized complexity of a version of house-marketing that allows for approximation.
Finding meaningful versions of \ra{} that allow for approximation is another alley for future research.

\bibliographystyle{plainurl}
\bibliography{reference}

\end{document}